\documentclass[11pt,a4paper]{article}
\usepackage[latin,english]{babel}
\usepackage[utf8]{inputenc}
\usepackage[T1]{fontenc}
\usepackage{eufrak}
\usepackage{geometry}
\usepackage{amsmath}
\usepackage{bbm}
\usepackage{amssymb}
\usepackage{amsbsy}
\usepackage{amsthm}
\usepackage{makeidx}
\usepackage{mathtools}
\usepackage{mathabx, mathrsfs, dsfont}
\usepackage{cases}
\usepackage{braket}
\usepackage[toc,page]{appendix}
\usepackage{dirtytalk}
\usepackage{pdfsync}
\usepackage{graphicx}
\usepackage{epstopdf}
\usepackage{todonotes}
\usepackage{graphicx}
\usepackage{fancyhdr}
\usepackage{math}
\usepackage{cite}
\usepackage{color}
\usepackage{subcaption}
\usepackage{bbm}
\usepackage[colorlinks=true, pdfstartview=FitV, urlcolor=blue, citecolor=red, linkcolor=blue]{hyperref}
\mathtoolsset{showonlyrefs}

\newcommand{\Volt}{\textrm{Volt}}
\newcommand{\trace}[1]{\mathrm{Tr}\left( #1\right)}

\newcommand{\dr}[1]{\left[#1\right]^{\textrm{dr}}}
\newcommand{\ve}{v_{\textrm{eff}}}

\numberwithin{equation}{section}

\title{Generalized Hydrodynamics for the Volterra lattice: ballistic and non-ballistic behavior of correlation functions}
\author{Guido Mazzuca\footnote{Tulane University, Louisiana, USA. 
\texttt{Mail}: gmazzuca@tulane.edu}}

\date{\today}

\begin{document}
	\maketitle

\begin{abstract}
	In recent years, a lot of effort has been put in describing the hydrodynamic behavior of integrable systems. In this paper, we describe such picture for the Volterra lattice. Specifically, we are able to explicitly compute the \textit{susceptibility matrix} and the \textit{current-field} correlation matrix in terms of the density of states of the Volterra lattice endowed with a Generalized Gibbs ensemble. Furthermore, we apply the theory of linear Generalized Hydrodynamics to describe the Euler scale behavior of the correlation functions. We anticipate that the solution to the Generalized Hydrodynamics equations develops shocks at $\xi_0=\frac{x}{t}$; so this linear approximation does not fully describe the behavior of correlation functions. Intrigued by this fact, we performed several numerical investigations which show that, exactly when the solution to the hydrodynamic equations develops shock, the correlation functions show an highly oscillatory behavior. In view of this empirical observation, we believe that at this point  $\xi_0$ the diffusive contribution are not sub-leading corrections to the ballistic transport, but they are of the same order.
\end{abstract}

        \section{Introduction}
        In recent years, a lot of effort has been put in describing the hydrodynamic behavior of integrable systems, i.e. dynamical system whose evolution can be explicitly computed in terms of the initial data. Specifically, it has been a big mathematical challenge to fully describe the correlation functions of such integrable models.  Recently, physicists have introduced a new theory that aims to describe the behavior of such functions, the so-called \textit{Generalized Hydrodynamics} (GHD) \cite{DoyonNone,GHD1,GHD2}{ \color{black}, see also the recent review \cite{GHD3} and the reference therein}. The underline idea of this theory is to obtain a set of hydrodynamic equations describing the macroscopic evolution of the considered medium; those equations also describes the evolution of the correlation functions. 
        
        Despite not being fully mathematically rigorous, using this theory H. Spohn was able to describe the behavior  of the correlation functions for the Toda lattice \cite{spohn2021hydrodynamicToda,Spohn2019a,Spohn2020}, see also \cite{Doyon2019,Bulchandani2019}. His results were confirmed by comparing the prediction of the generalized hydrodynamics with numerical simulations, see \cite{mazzuca2023equilibrium}. 
        
         H. Spohn was able to carry out this computation relaying on results from \textit{Random Matrix theory} (RMT). In particular, {\color{black} enforcing the relation between the Toda lattice and the classical Real Beta ensemble \cite{Dumitriu2002,mazzuca2021mean},} he was able to  describe the linear approximation of the correlation functions of the Toda lattice. {\color{black}This relation between the Toda lattice and the Real beta ensemble is not unique.} Indeed, after Spohn breakthrough, several authors enforced this idea in order to describe statistical properties of the dynamical systems at hand. For example,  in \cite{mazzuca2021generalized}, the authors were able to describe the density of states of the Ablowitz-Ladik lattice in terms of the one of the circular $\beta$ ensemble \cite{KillipNenciu2004}, independently Spohn obtained an analogous result \cite{spohn2021hydrodynamic}. In \cite{Guionnet2022,Mazzuca2023}, the authors obtained a large deviation principle linking the Toda  lattice and the Ablowitz-Ladik lattice with the Real $\beta$ ensemble and the Circular $\beta$ ensemble respectively. Another interesting result in this direction is \cite{Grava2023}, in this paper the authors established connections between the classical Gibbs ensemble for the Exponential Toda lattice and the Volterra lattice with the Laguerre ensemble \cite{Dumitriu2002} and the Antisymmetric $\beta$-ensemble \cite{Dumitriu_Forrester}, respectively. We want to mention also of the work \cite{Grava2021Correlation}, where the authors computed explicitly the correlation function for the short range harmonic chain; they were also able to describe the long time asymptotic of those correlations in great details.

         {\color{black}
          GHD has a wide range of applications. For example, it can be applied to several \textit{quantum integrable models} to compute their correlation functions and, in general, their thermodynamic behavior \cite{GHD2,QuantumGHD1,Bastianello2,QuantumGHD2}. Furthermore, it is possible to fully describe the thermodynamics of several classical integrable models by means of semiclassical limit of their quantum counterpart, which can be analyzed through GHD \cite{Bastianello1,Bastianello3}.
         GHD has also been an important tool to describe the soliton gas picture for several integrable PDE models, see \cite{ElWhithamEq,El2005,El2010,Doyon2018,Girotti2021,Bonnemain2022}. For a wider introduction to the topic, we refer to the recent reviews  \cite{Soliton1,Soliton2} and the references therein. }
         
        In this paper,  we consider the \textit{Volterra lattice} \cite{Moser75} and we compute the \textit{susceptibility matrix} and the \textit{current-field correlation matrix}. Furthermore, we apply the theory of Generalized hydrodynamics to describe the Euler scale behavior of the correlation functions.  We anticipate that the solution to the differential equations describing the Euler scale dynamics \textit{develops shock} for some explicit value $\xi_0=\frac{x}{t}$. { \color{black}The reason for this shock is that the effective velocity of the system $v_{\text{eff}}(\lambda) $ has a minimum at $0$. Such behavior was also observed in other quantum systems, see \cite{QuantumGHD1,Shock1}, and has also some analogy with the classical chain of couple oscillator, see \cite{mazzuca2023equilibrium}. In the later case, the correlation functions have different scaling properties around the minimum point. We mention that a new approach to describe such shock formation through Witham modulation theory was recently developed in \cite{Whitham1}, and this approach might be extended to our case. }
        
        \begin{figure}
\centering
\includegraphics{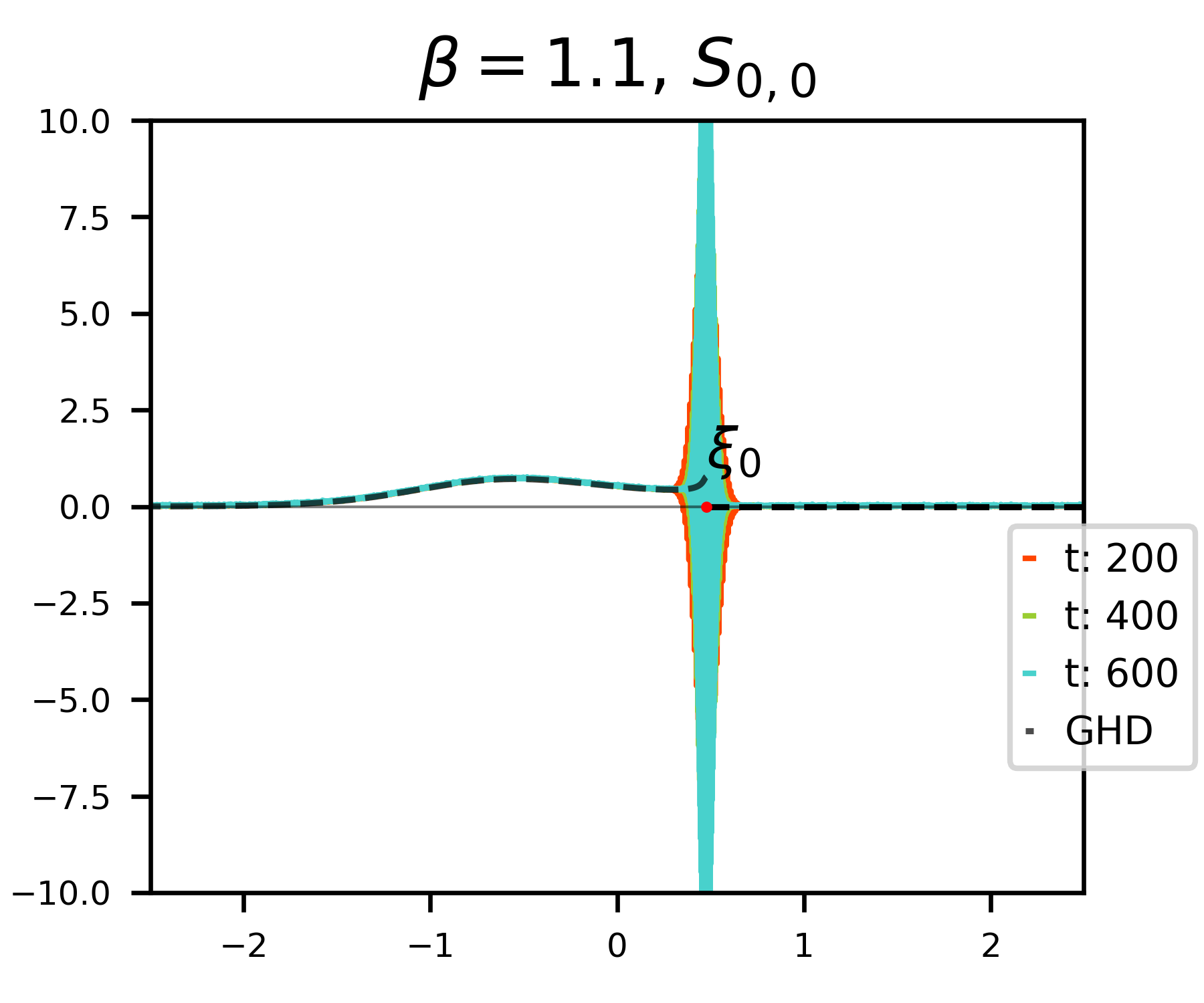}
\includegraphics{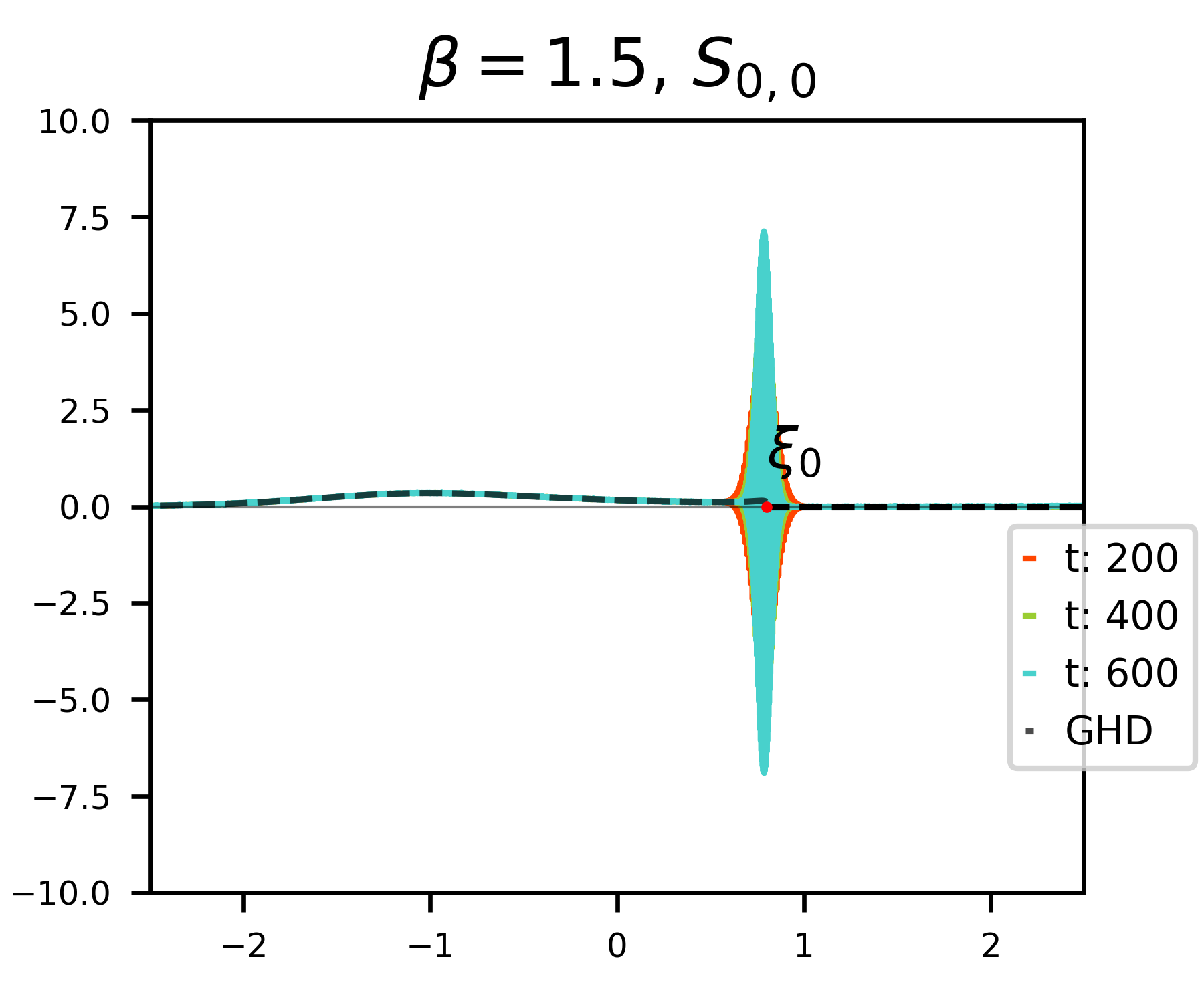}
\includegraphics{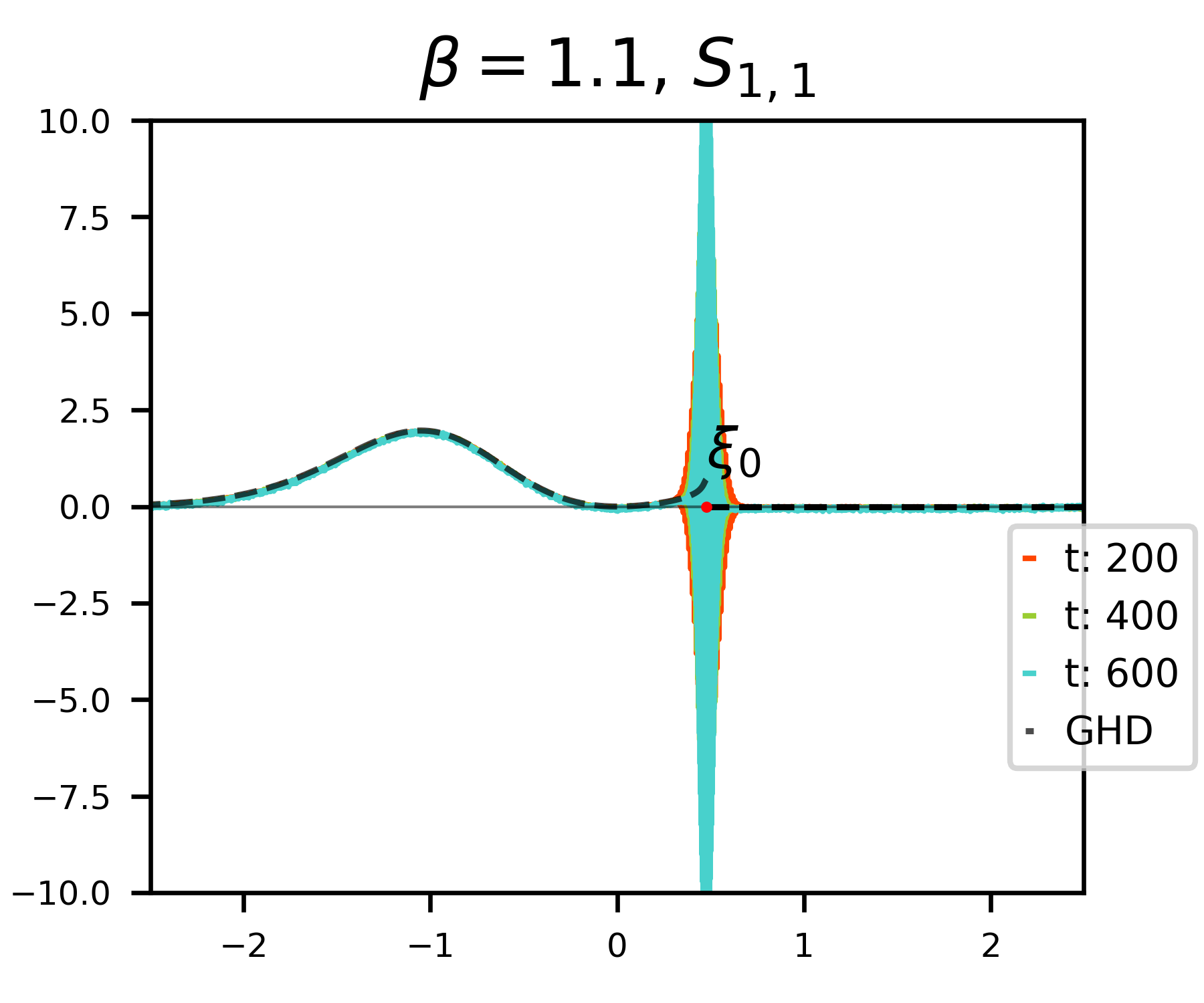}
\includegraphics{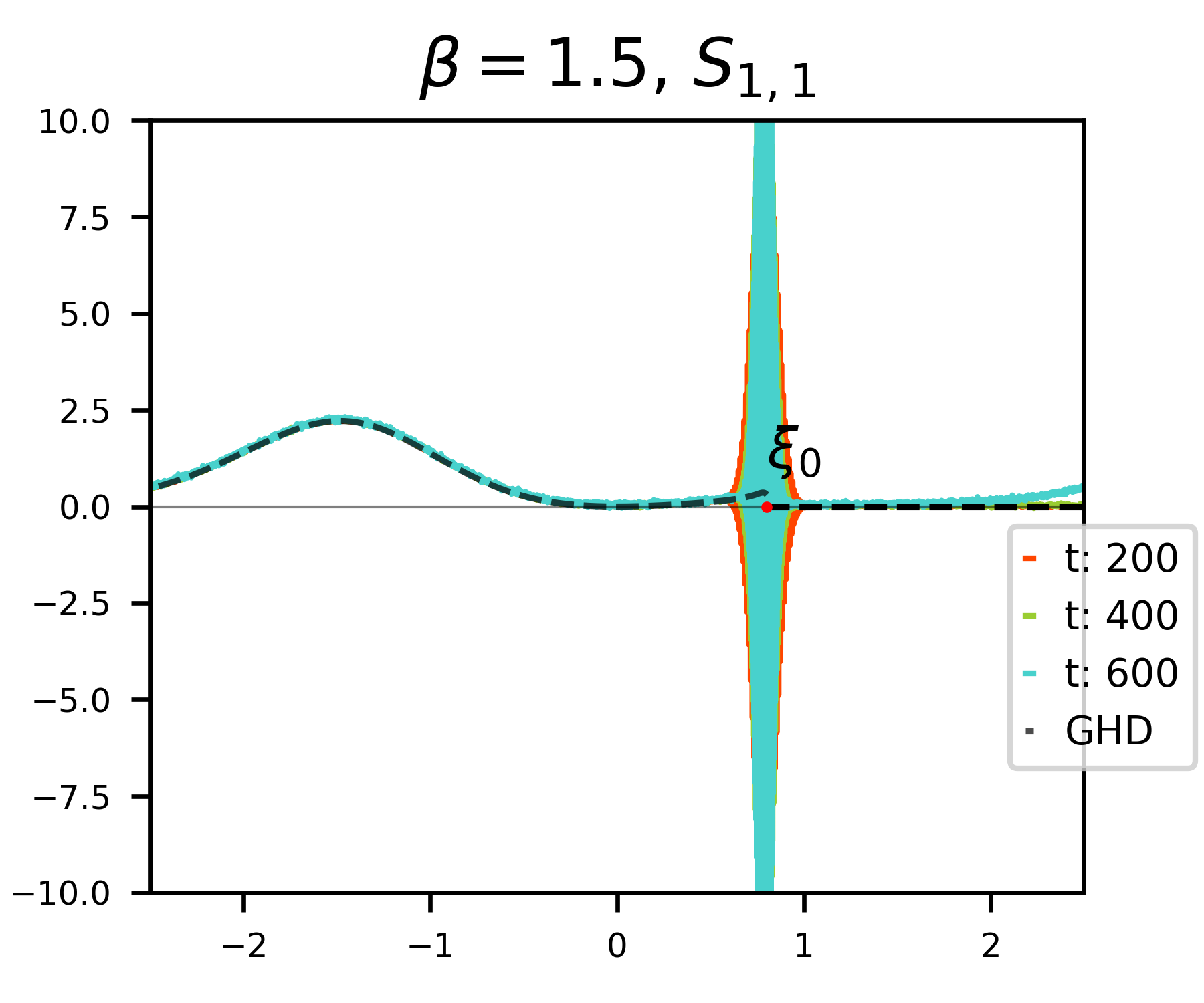}
\includegraphics{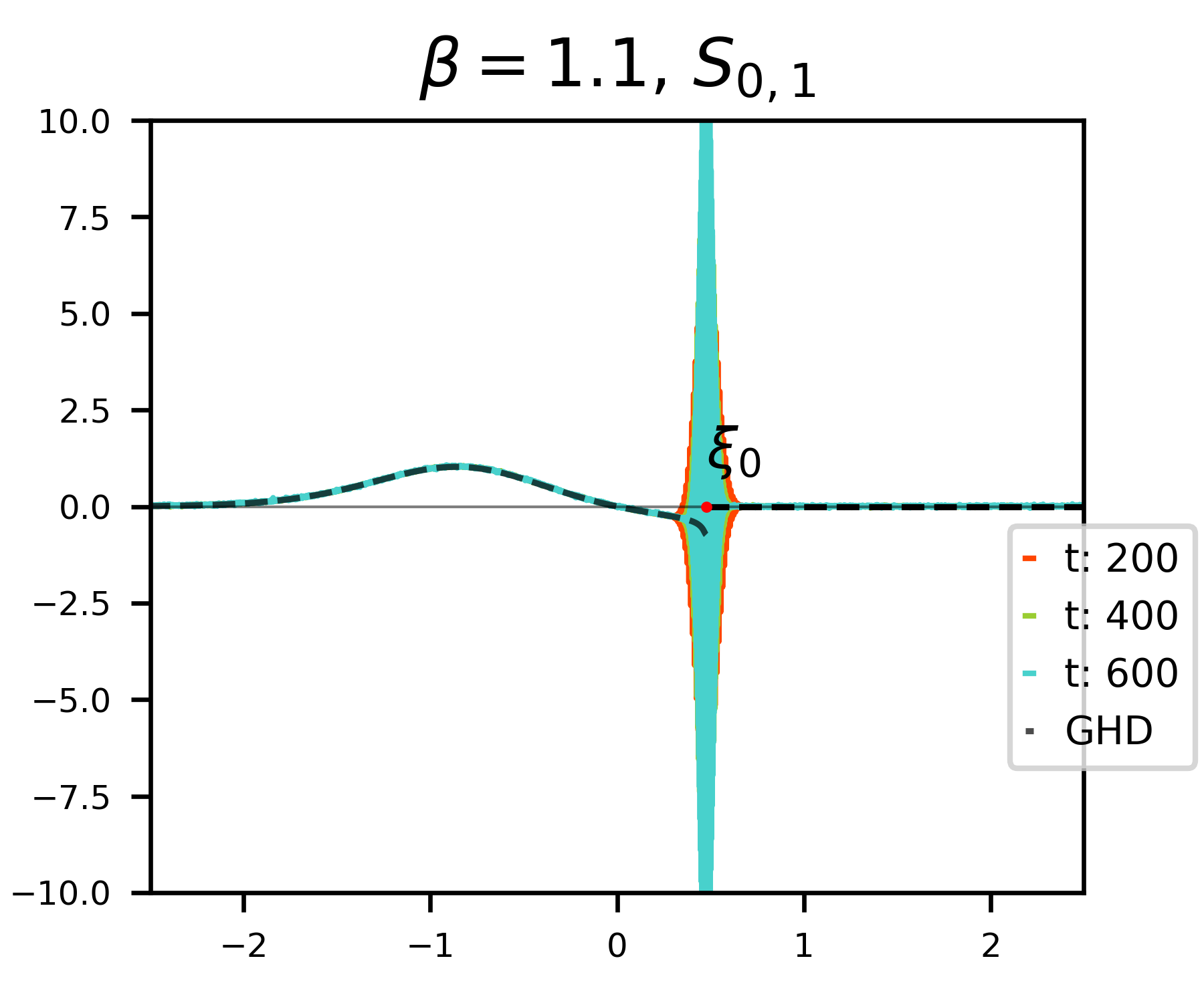}
\includegraphics{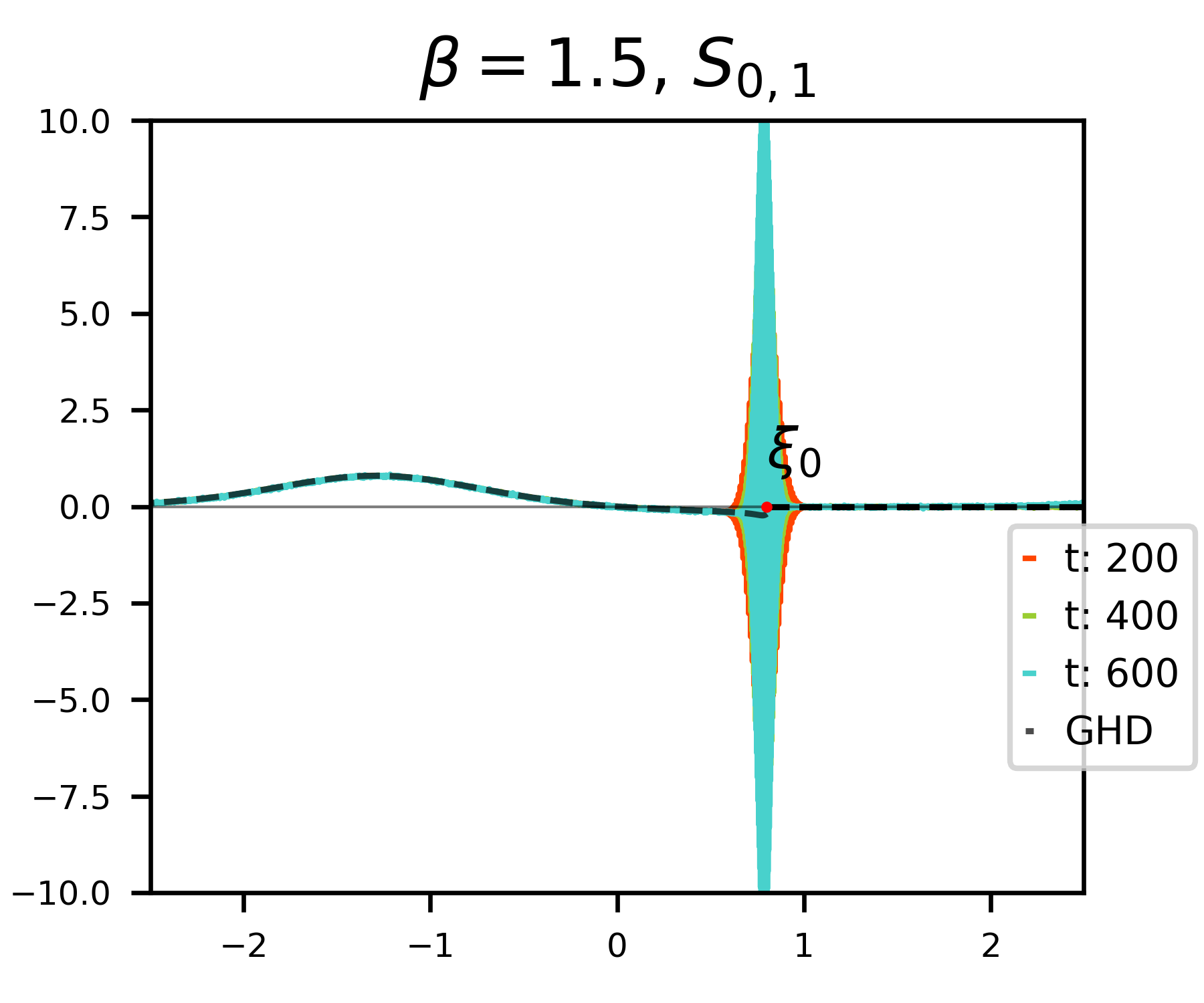}
\caption{Volterra correlation functions: GHD predictions vs molecular dynamics simulation. Left panels: number of particles: $3000$, trials: $10^6$, $\beta=1.1$. Right panels: number of particles: $3000$, trials: $10^6$, $\beta=1.5$. The variable on the $x$-axis is $\xi=\frac{x}{t}$. }
\label{fig:paragone1}
\end{figure}

Intrigued by {\color{black} the shock formation}, we perform several numerical experiments to  investigate the behavior of the correlation functions, the results are displayed in Figure \ref{fig:paragone1}. {\color{black} After the usual ballistic scaling, we notice that at the point $\xi_0$ where the GHD predictions - displayed with dashed lines - develop shock, the correlation functions show an highly oscillatory behavior with a different scaling, exactly as in the case of the harmonic chain.  }

        {\color{black}
The manuscript is organized as follows. In section \ref{sec:model_res}, we introduce our model of interest, i.e. the Volterra lattice, and state our results. In section  \ref{sec:TheoBack}, we present the theoretical framework that we exploit to compute the susceptibility matrix $C$ and the charge-current correlation matrix $B$ (Theorem \ref{thm:main}). Specifically, we recall the results in \cite{Mazzuca2024CLT}, and we used them to compute the susceptibility matrix $C$ and the mean value of the currents $J^{[k]}$ in terms of the free energy \eqref{eq:free_energy_volterra}; furthermore, we formally describe the density of states of the model. In section \ref{sec:antisym}, we introduce the Antisymmetric Gaussian $\beta$ ensemble in the high temperature regime; this is a random matrix ensemble introduced by \cite{Dumitriu_Forrester}, we enforce several results related to this ensemble in order to prove Theorem \ref{thm:main}. In section \ref{sec:computationCB}, we prove Theorem \ref{thm:main}, i.e. we explicitly compute the susceptibility matrix $C$ and the charge-current correlation matrix $B$ in terms of the density of states of the Volterra lattice. This is the main analytic contribution of our paper. In section \ref{sec:linearhydro}, we apply the theory of generalized hydrodynamics to obtain the linear order approximation of the correlation functions for the Volterra lattice. Finally, in section \ref{sec:numerical_results}, we describe the numerical results that we obtained and the procedure that we applied.}
         
         \section{Model and Results}
         \label{sec:model_res}
        	The \textit{Volterra lattice}, also known as the \textit{discrete KdV equation}, describes the  motion of $2N$ particles on the line with equations 
	\begin{equation}   
 \label{Volterra}
	\frac{d }{d t} a_j \equiv	\dot{a_j} = a_j \left(a_{j+1} - a_{j-1} \right), \qquad j=1,\dots,2N.
	\end{equation}
	It was originally  introduced  by Volterra to study 
	population evolution in a hierarchical system of competing species. It was first solved by Kac and Van Moerbeke in \cite{Kac1975} using a discrete version of inverse scattering transform due to Flaschka \cite{Flashka1974b}. Equations \eqref{Volterra} can be considered as a finite-dimensional approximation of the Korteweg–de Vries equation.
	The phase space is $\mathbb{R}_+^{2N}$ and we consider  periodic boundary conditions $a_j=a_{j+2N}$ for all $j\in\Z$. The Volterra lattice  is a  reduction of the \emph{second flow of the Toda lattice} \cite{Kac1975}. Indeed, the latter is described by the dynamical system
	\begin{align}
		\dot{a_j} &= a_j \left(b_{j+1}^2 - b_{j}^2 + a_{j+1}-a_{j-1} \right),     &   &\qquad j=1,\dots,2N,    \\
		\dot{b_j} &= a_j(b_{j+1}+b_j) -a_{j-1}(b_j+b_{j-1}),       &   &\qquad j=1,\dots,2N,     \label{2ndToda}
	\end{align}
	and equations \eqref{Volterra} are recovered just by setting $b_j\equiv0$. The Hamiltonian structure of the equations follows from the one of the Toda lattice.  On the phase space $\R^{2N}_+$  we introduce the   Poisson bracket
	\begin{equation}
		\label{eq:poisson_volterra}
		\{ a_j, a_i \}_{\Volt} =  a_ja_i(\delta_{i,j+1} - \delta_{i,j-1})\,
	\end{equation}
	and the Hamiltonian  $H_1 = \sum_{j=1}^{2N} a_j\,$ so that the equations of motion \eqref{Volterra}  can be written in the Hamiltonian form 
	\begin{equation}
		\dot{a}_j = \{ a_j, H_1\}_{\Volt}\,.
		\label{eq:hamvoltN}
	\end{equation}
	An elementary constant of motion for the system is $H_0 = \prod_{j=1}^{2N} a_j$ which is  independent of  $H_1$.
	The Volterra lattice  is a completely integrable system, and  it admits several equivalent \emph{Lax representations}, see e.g. \cite{Kac1975,Moser75}. The classical one reads
	\begin{equation}  
		\dot{L}_1 = \left[A_1,L_1\right],
	\end{equation}
	where 
	\begin{equation}
		\begin{split}
			\label{eq:classic_Vlax}
			L_1 &= \sum_{j=1}^{2N} a_{j+1}E_{j+1,j}+E_{j,j+1},             \\
			A_1 &= \sum_{j=1}^{2N} (a_j + a_{j+1}) E_{j,j} + E_{j,j+2}\,,     
		\end{split}
	\end{equation}
	where we define the matrix $E_{r,s}$  as $\left(E_{r,s}\right)_{ij}=\delta^i_r \delta^j_s$ and   $E_{j+2N,i} =E_{j,i+2N} = E_{j,i}$.
	There exists also a \emph{symmetric} formulation due to Moser \cite{Moser75},
	\begin{equation}
		\label{Vlax}
		\begin{split}
			\dot{L}_2 & = \left[A_2,L_2\right]\\
			L_2 &= \sum_{j=1}^{2N} \sqrt{a}_j (E_{j,j+1}+E_{j+1,j})\,,             \\
			A_2 &= \frac{1}{2}\sum_{j=1}^{2N} \sqrt{a_ja_{j+1}} (E_{j,j+2}-E_{j+2,j})\,,     
		\end{split}
	\end{equation}
	which assumes that all $a_j > 0$.
	Furthermore, as it was noticed in \cite{Grava2023}, there exists also an \textit{antisymmetric} formulation for this Lax pair, indeed a straightforward computation yields 
	\begin{proposition} \label{AntiProp}
		Let $a_j > 0$ for all $j=1,\ldots,2N$. Then, the dynamical system \eqref{Volterra} admits an antisymmetric Lax matrix $L_3$ with companion matrix $A_3$, namely the equations of motion are equivalent to $\dot{L_3} = \left[A_3,L_3\right]$ with
		\begin{align}
			\label{LaxVolterra}
			L_3 &= \sum_{j=1}^{2N}\sqrt{ a_j} (E_{j,j+1}-E_{j+1,j}),             \\
			A_3 &= \frac{1}{2}\sum_{j=1}^{2N} \sqrt{a_ja_{j+1}} (E_{j+2,j}-E_{j,j+2}).      \label{AntiVlax}
		\end{align}
	\end{proposition}

In view of the Lax representation $L_3 \equiv L$, we deduce that $\{Q^{[n]}\} _{n=1}^N= \{(-1)^n\trace{L^{2n}}\}_{n=1}^{N}$ are constants of motion for the system, or \textit{conserved field}, i.e. $\frac{\di }{\di t}\trace{L^{2n}} = 0$. We notice that for $k\in \N$ $\trace{L^{2k+1}}\equiv 0$ in view of the antisymmetric property of the matrix $L$, and that $2H_1 = Q^{[1]}$.
 
 Since also $H_0$ is conserved, we define

 \begin{equation}
     Q^{[0]} = \frac{1}{2}\ln(H_0)\,,
 \end{equation}
    and the \textit{local conserved fields} $Q^{[n]}_j$  $j=1,\ldots, N$ as

    \begin{equation}
    \label{eq:locally_conserved_fields}
        Q^{[n]}_j =(-1)^n L^{2n}_{j,j} \quad n=1,\ldots,N\,, \quad Q^{[0]}_j =\frac{1}{2} \ln(a_j)\,.
    \end{equation}

%
%\todo[inline]{Sposare altrove, qui non mi serve questa specifica informazione}
%

 To compute the correlation functions, we must consider the \textit{currents} related to the locally conserved field, specifically

 \begin{equation}
 \label{eq:evolution_local_fields}
     \frac{\di}{\di t} Q^{[n]}_j = \frac{(-1)^n}{2}\left(L^{2n}_{j,j+2}\sqrt{a_ja_{j+1}} - L^{2n}_{j,j-2}\sqrt{a_{j-1}a_{j-2}}\right)\,,  \quad n = 1,\ldots, N 
 \end{equation}
 thus defining

 \begin{equation}
 \label{eq:currents}
     J_j^{[n]} = \frac{(-1)^n}{2} \left(L^{2n}_{j,j+2}\sqrt{a_ja_{j+1}} +  L^{2n}_{j-1,j+1}\sqrt{a_{j-1}a_{j}}\right)\,,  \quad n = 1,\ldots, N 
 \end{equation}
 we can rewrite  \eqref{eq:evolution_local_fields} as

 \begin{equation}
     \frac{\di}{\di t} Q^{[n]}_j = J_j^{[n]} - J_{j-1}^{[n]}\,, \quad n = 1,\ldots, N 
 \end{equation}
 we notice that $ L^{2n}_{j-1,j+1}\sqrt{a_{j-1}a_{j}}$ is  a boundary term, that allows us to write \eqref{eq:evolution_local_fields} in a compact form.

 For $n=0$ we can define

 \begin{equation}
     J_j^{[0]} =\frac{1}{2}( a_{j+1} + a_j)\,,
 \end{equation}
 and we can cast the evolution for $Q_j^{[0]}$ as

 \begin{equation}
        \frac{\di}{\di t} Q^{[0]}_j = J_j^{[0]} - J_{j-1}^{[0]}\,.
 \end{equation}

 \begin{remark}
 \label{rem:trick_rem}
     We notice that 

     \begin{equation}
         J_j^{[0]} =  - \frac{1}{2}Q_{j+1}^{[1]}\,,
     \end{equation}
     thus also $J_j^{[0]}$ is a locally conserved field.
 \end{remark}

 Analogously to the conserved fields, we define the \textit{total current} $J^{[n]} = \sum_{j=1}^N J_j^{[n]}$ for $n=1,\ldots,N$.  
 
	\subsection{Generalized Gibbs Ensemble}
	
	We  introduce the generalized Gibbs ensemble for the Volterra lattice \eqref{Volterra} following \cite{Grava2023,Mazzuca2024CLT} as
	\begin{equation}
		\label{eq:VolterraGibbs}
		\di\mu_{\Volt}(\ba) = \frac{1}{Z_N^{\Volt,1}(\beta,V)}\prod_{j=1}^{2N} a_j^{\beta/2 - 1}\mathbbm{1}_{a_j>0} e^{\trace{V(L(\ba)}} \di  \ba  ,\quad \beta>0,
	\end{equation}
	where $V\, :\,\R \to \R$ is a polynomial of the form $V(x) = (-1)^{\ell+1} c_{\ell}x^{2\ell} + \sum_{j=1}^{\ell-1} c_j x^{2j}$, $\ell\geq 1$, $c_{\ell} > 0$, and 
	\begin{equation}
		Z_N^{\Volt,1}(\beta,V) = \int_{\R^{2N}_+}\prod_{j=1}^{2N} a_j^{\beta/2 - 1}\mathbbm{1}_{a_j>0}e^{\trace{V(L)}}\di \ba <\infty\,.
		\label{eq:ZNvolterravalue}
	\end{equation}
	
	We recover the standard Gibbs ensemble setting $V(x) = x^2/2$, in this case the variables $a_j$ are independent and identically distributed according to a random variables with  probability density function $f_\beta(x)$
	
	\begin{equation}
		f_\beta(x) = \frac{x^{\beta/2-1} e^{-x}}{\Gamma(\beta/2)} \,,
	\end{equation}
	
	which is just a scaled $\chi^2$ distribution with parameter $\beta$.
In this case, the partition function can be computed explicitly:

\begin{equation}
	Z_N^{\Volt,1}(\beta, x^2/2) = \Gamma(\beta/2)^{2N}\,.
\end{equation}
{\color{black}
\begin{remark}
We recall that one can express the conserved fields as  $Q^{[k]} = (-1)^k\trace{L^{2k}}$ for $k=1,\ldots,N$ and that $Q^{[0]} = \frac{1}{2}\sum_{j=1}^{2N}\ln(a_j)$; therefore, one can rewrite the previous GGE in the form

\[ \di\mu_{\Volt}(\ba) = \frac{1}{Z_N^{\Volt,1}(\beta,V)} \exp\left( \left(\beta -2\right)Q^{[0]} -c_\ell Q^{[\ell]} + \sum_{j=1}^{\ell-1} (-1)^jc_j Q^{[j]}\right) \di\ba.\]

We stress that the $-2$ in the previous expression comes from the volume form related to the Poisson brackets \eqref{eq:poisson_volterra} 

\end{remark}}
	For future computations, it is useful to represent the previous expressions in terms of the variables $\{x_j\}_{j=1}^{2N}$ defined as $x_j^2 = a_j$, such that $x_j\in \R_+,\, j=1, \dots 2N$. In this new set of variables, we can express the Gibbs measure \eqref{eq:VolterraGibbs} and its normalization \eqref{eq:ZNvolterravalue} as  
	
	\begin{equation}
		\label{eq:quadratic_Gibbs}
		\begin{split}
			&\di\mu_{\Volt}(\bx) = \frac{1}{Z_N^{\Volt,2}(\beta,V)}\prod_{j=1}^{2N} x_j^{\beta - 1}\mathbbm{1}_{x_j>0} e^{\trace{V(L(\bx^2))}} \di  \bx  \\
			& Z_N^{\Volt,2}(\beta,V) = \int_{\R^{2N}_+}\prod_{j=1}^{2N} x_j^{\beta - 1}\mathbbm{1}_{x_j>0}e^{\trace{V(L(\bx^2))}}\di \bx <\infty
		\end{split},\quad \beta>0,
	\end{equation}
where we defined $\bx^2 = (x_1^2. \ldots,x_{2N}^2)$. Furthermore, we notice that in the case $V(x) = \frac{x^2}{2}$, the random variables $x_j$ are distributed as $2N$ independent  $\chi$-distributions, i.e. their probability density function is of the form  $g_\beta(x)$

\begin{equation}
	g_\beta(x) = \frac{x^{\beta - 1}e^{-x^2}}{2^{-1}\Gamma\left(\frac{\beta}{2}\right)}\,.
\end{equation}
In this case is possible to compute the partition function as

\begin{equation}
\label{eq:partition_easy_volterra}
	Z_N^{\Volt,2}(x^2/2,\beta) = 2^{-2N}\Gamma\left(\frac{\beta}{2}\right)^{2N}\,.
\end{equation}

In this new coordinates, the Lax matrix $L$ \eqref{LaxVolterra} reads as

\begin{equation}
	L= \sum_{j=1}^{2N}x_j (E_{j,j+1}-E_{j+1,j}),  
\end{equation}

The \textit{main analytic result} of this paper is the explicit computation of the susceptibility matrix $C\in \textrm{Mat}(\R,N+1)$ and the charge-current static correlation matrix $B\in\textrm{Mat}(\R,N+1)$ in terms of the \textit{density of states} $\sigma_{\beta,V}$ of the Volterra lattice endowed with the probability distribution $\di \mu_\Volt$ \eqref{eq:VolterraGibbs}. The matrices $C,B$ are defined as 

\begin{equation}
	\label{eq:defCB}
	C_{m,n} = \lim_{N\to\infty}\frac{1}{2N}\cov{Q^{[n]}}{Q^{[m]}} \,,\quad B_{n,m} = \lim_{N\to\infty}\frac{1}{2N}\cov{Q^{[n]}}{J^{[m]}}\,,\quad m,n=0,\ldots,N\,,
\end{equation} 
where the covariance $$\cov{f}{g} = \meanvalue{fg}{1}- \meanvalue{f}{1}\meanvalue{g}{1}\,,$$ $\meanvalue{\cdot}{1}$ is the expected value with respect to the GGE \eqref{eq:VolterraGibbs} and we adopt the convention that if the quantity is evaluated at time $t=0$, we omit the time dependence. We notice that in \cite{Mazzuca2024CLT}, the authors showed how to compute the correlation matrix $C$ \eqref{eq:defCB} for the Volterra lattice in terms of the Free energy \eqref{eq:free_energy_volterra}. 

{\color{black} To get an explicit expression of $C,B$, we must obtain the \textit{density of states} of the Volterra lattice.} The density of state $\sigma_{\beta,V}$ is the probability distribution on $\R$ defined as the weak limit of 

\begin{equation}
	\lim_{N\to\infty}\sum_{j=1}^{2N} \delta_{w_j}(x) \rightharpoonup \sigma_{\beta,V}\,,
\end{equation}
where $-iw_j$ are the  eigenvalues of the Lax matrix $L$ \eqref{LaxVolterra}, and $\delta_y(x)$ is the delta function centered at $y$.  

{\color{black} Given these definitions, we can state the main analytical result of our paper:}

\begin{theorem}
	\label{thm:main}
	Consider the Lax matrix $L$ \eqref{LaxVolterra} endowed with the GGE \eqref{eq:VolterraGibbs}. Define the susceptibility matrix $C$ and the charge-current correlation matrix $B$ as in \eqref{eq:defCB}. Then,
	
	\begin{equation}
		\begin{split}
				& C_{0,0} =  \frac{\kappa^2}{2}\jap{\sigma_{\beta,V}(\dr{1})^2}\,, \\
				& C_{0,n} =C_{n,0}=   \kappa \jap{\sigma_{\beta,V}\dr{1}\left(\dr{w^{2n}} -q_n\dr{1}\right)} \,, \quad n= 1, \ldots,N\\
				&  C_{m,n} =   2 \jap{\sigma_{\beta,V}\left(\dr{{w^{2m}}} - q_m[1]^{dr}\right) \left(\dr{{w^{2n}}} - q_n[1]^{dr}\right)} \,\quad m,n = 1,\ldots,N\,,\\
		& B_{0,n} =B_{n,0}=  -\frac{1}{2} C_{n,1}\,, \quad n= 0, \ldots,N\\
		&  B_{m,n}  =  -\frac{2}{\kappa}\jap{\sigma_{\beta,V} (v_{\textrm{eff}} - q_1) \dr{{w^{2m}} -q_m} \dr{{w^{2n}} - q_n}}\,
	\end{split}
\end{equation}

Here $\jap{\phi} = \int_\R \phi(x)\di x$, 
$\sigma_{\beta,V} = \partial_{\beta}(\beta \rho_{\beta,V})$, where the derivative is understood in week sense, and  
$\rho_{\beta,V}$ is the minimizer of the following functional 

	\begin{equation}
	\label{eq:functional_intro}
	\begin{split}
		\cF(\beta,V)[\rho] &		= -\frac{\beta}{2} \int\int_{\R_+^2}\ln(\vert x^2-y^2 \vert)  \rho(x) \rho(y)\di x \di y - \int_{\R_+}\left(V(ix)  + V(-ix)  - \ln(x)\right) \rho(x)\di x \\ & +\int_{\R_+}\ln( \rho(x)) \rho(x)\di x \,,
	\end{split}
\end{equation}
here, $V(x)$ is a polynomial of the form $V(x) = (-1)^{\ell+1} c_{\ell}x^{2\ell} + \sum_{j=1}^{\ell-1} c_j x^{2j}$, $\ell\geq 1$, $c_{\ell} > 0$.
The dressing operator $\dr{\psi}$ is defined as

\begin{equation}
	\label{eq:dress_T_def}
	\dr{\psi} =  (1-\beta T\rho_{\beta,V})^{-1}\psi, \quad  T\psi(w) = \int_{\R}\ln(\vert w^2-z^2\vert )\psi(z)\di z\,,\quad w\in \R\,.
\end{equation}

$q_m$ is the $2m^{th}$ moment of $\sigma_{\beta,V}$, i.e. $q_m = \int_{\R_+} \sigma_{\beta,V}(w)w^{2m}\di w $, $\ve= \frac{\dr{w^2}}{\dr{1}}$ and 

\begin{equation}
	\kappa = \partial_\beta 2\cF_{\textrm{Volt}}(\beta,V)\,,  \quad \cF_{\textrm{Volt}}(\beta,V) = -\lim_{N\to \infty} \frac{1}{2N} \ln\left(Z_N^{\Volt}(\beta,V)\right)\,.
\end{equation}
\end{theorem}

\begin{remark}
	From the explicit expression of $C,B$ the two matrices are symmetric, this is obvious  for $C$ due to its structure, but it is not for $B$. 
\end{remark}

The explicit computation of these two matrices allows us to apply the theory of Generalized Hydrodynamics and deduce the behavior of the space-time correlation functions at the Euler scale. Specifically, we argue that defining 

\begin{equation}
	S^{N}_{m,n}(j,t) = \lim_{N\to \infty} \cov{Q^{[m]}_j(t)}{Q^{[n]}_0(0)} \,,
\end{equation} 
its approximation at the Euler scale is 

\begin{equation}
	\lim_{N\to \infty} S^{N}_{m,n}(j,t) \stackrel{x= \frac{j}{2N}}{\simeq} \tS_{m,n}(x,t)\,,  
\end{equation}

where

\begin{equation}
	\label{eq:correlation}
	\tS_{m,n}(x,t) = \begin{cases}
		 \vspace{10pt}
		 \frac{\kappa^2}{2}\jap{\sigma_{\beta,V} \delta(x+t(\ve - q_1)(\kappa)^{-1})\left((\dr{1})^2  \right)}  \quad m,n = 0\,, \\
		 \vspace{10pt}
		  \langle \kappa \sigma_{\beta,V} \delta(x+t(\ve - q_1)(\kappa)^{-1})\dr{1} \Xi[w^n]\rangle \quad m=0\,, \\
		  \vspace{10pt} 
		 \langle \kappa \sigma_{\beta,V} \delta(x+t(\ve - q_1)(\kappa)^{-1})\dr{1} \Xi[w^m]\rangle \quad n=0\,, \\ 
		 	 2\langle\Xi[w^m]\Xi[w^n]\sigma_{\beta,V}  \delta(x+t(\ve - q_1)(\kappa)^{-1}) \rangle \quad m,n\neq 0\,,
	\end{cases} \,. 
\end{equation}

and $\Xi \phi = \dr{\phi - \jap{\sigma_{\beta,V} \phi}}$.
{\color{black}
\begin{remark}
We notice that in \cite{Doyon_large_scale} the Euler scale correlation functions have a Dirac delta function of the form $\delta(x-\tv_\text{eff} t)$, one can recover such expression by the change of coordinates $ \tv_\text{eff} = - (\ve - q_1)(\kappa)^{-1}$. Indeed, the two effective velocities are linearly related.
\end{remark}}

{\color{black} As we already mentioned}, the function $\tS_{m,n}(x,t)$ \textbf{is not continuous} for all $(x,t)$; the two main reasons are that the density $\sigma_{\beta,V}$ has support just on the positive real axis, {\color{black} and that the effective velocity $\ve(w)$ has a minimum at $0$, therefore it is not one-to-one form $\R_+$ to $\R$.} Intrigued by this fact, we performed several numerical investigation to compare the numerical correlation functions and the prediction obtained from the linearized Generalized Hydrodynamic (GHD). We extensively analyze them in the last section of our paper, here we summarize our findings

\begin{itemize}
	\item The GHD correctly predict the \textbf{ballistic scaling } of the correlation functions; i.e.
	 $$\lim_{N\to \infty} S^{N}_{m,n}(j,t) \sim \frac{1}{t}f\left( \frac{x-\tv t}{\tc t}\right)$$ for some function $f$ and constant  $\tv,\tc$. 
	 \item The approximation $\tS_{m,n}(x,t)$ is not continuous for $\frac{x}{t} = \xi_0 = \frac{ q_1 - \ve(0) }{\kappa}$
	 \item In a space-time neighbor of the points $(x,t)$ such that  $\xi_0 = \frac{x}{t}$, i.e.  where $\tS_{m,n}(x,t)$ is not continuous, the numerical correlation functions display an \textbf{highly oscillatory behavior}.
\end{itemize}

The combination of these facts lead us to believe that, in order to obtain a more accurate prediction, one has to consider also some diffusive effect as described in \cite{DeNardis2018}. Specifically, we believe that at the point $\xi_0$ the diffusive effects are not a sub-leading correction to the transport dynamics, but they are of the same order. 
We notice, that this is not the first time that such effect has been noticed, see \cite{Piroli2017,Misguich2017,Ljubotina2017,Shock1}. Nevertheless, up to our knowledge, this is the first time that such behavior is present in a classical nonlinear integrable chain at equilibrium. {\color{black} As we already mentioned, in \cite{Whitham1} the authors proposed a new approach to study this shock formation through Whitham modulation theory; it is possible that such methodology could also be applied in this context.}

{\color{black}
\subsection{Idea of the proof}

The proof follows the same approach developed by Spohn in the series of papers \cite{Spohn2019a,Spohn2020,spohn2021hydrodynamic,spohn2021hydrodynamicToda}. The first step is to express both the average conserved fields and the average currents in terms of the free energy of the model. To do that, we use the results in \cite{Mazzuca2024CLT}. Then, we must express these quantities in terms of the density of states of the model, and we have to characterize it. We do this by linking the Generalized Gibbs ensemble of the Volterra lattice with the classical Antisymmetric beta ensemble at high temperature \cite{Dumitriu_Forrester,Mazzuca2021}. 
Using this connection, we are able to compute explicitly both the susceptibility matrix $C$ and the charge-current correlation matrix $B$. 
Given these two matrices, it is possible to apply GHD to obtain an approximation of the correlation functions of the model at hand.  
}

\section{Theoretical  Framework}
\label{sec:TheoBack}

In this section, we recall several known results that we use to prove Theorem \ref{thm:main}. In particular, we use the results in \cite{Mazzuca2024CLT,Guionnet2022}.
 
\subsection{Average Conserved fields}

        In \cite{Mazzuca2024CLT}, the authors were able to compute the susceptibility matrix $C$ \eqref{eq:defCB} in terms of the \textit{ free energy} of the model, which is defined as 
        
        \begin{equation}
        	\label{eq:free_energy_volterra}
        	\cF_{\textrm{Volt}}(\beta,V) = -\lim_{N\to \infty} \frac{1}{2N} \ln\left(Z_N^{\Volt,2}(\beta,V)\right)\,.
        \end{equation}
         Specifically, they were able to prove the  following
        
    \begin{corollary}[cf. \cite{Mazzuca2024CLT}, Corollary 3.13]   
    \label{cor:mean_cor_gibbs}
    Consider $Q_j^{[n]}$ \eqref{eq:locally_conserved_fields}, the Generalized Gibbs ensemble $\di\mu_{\Volt}(\ba)$ \eqref{eq:VolterraGibbs}, and the free energy $\cF_{\textrm{Volt}}(\beta,V)$ \eqref{eq:free_energy_volterra}. For any fixed $n,m\in \mathbb{N}$ the following holds true

    \begin{equation}
    \begin{split}
         &\lim_{N\to\infty}\frac{1}{2N}\meanvalue{Q^{[n]}}{1} = -i\partial_t \cF_{\textrm{Volt}}(\beta,V+ (-1)^{n+1} itx^{2n})_{\vert_{t=0}}\,, \\
         &\lim_{N\to\infty}\frac{1}{2N}\meanvalue{Q^{[0]}}{1} = -\partial_\beta \cF_{\textrm{Volt}}(\beta,V)\, \\ 
         &\lim_{N\to\infty}\frac{1}{2N}\cov{Q^{[n]}}{Q^{[m]}} = \partial_{t_1}\partial_{t_2} \cF_{\textrm{Volt}}(\beta,V+  (-1)^{n+1}it_1x^{2n}+  (-1)^{m+1} it_2 x^{2m})_{\vert_{t_1=t_2=0}}\,,  \\
         &\lim_{N\to\infty}\frac{1}{2N}\cov{Q^{[n]}}{Q^{[0]}} = -i\partial_{t}\partial_{\beta} \cF_{\textrm{Volt}}(\beta,V + (-1)^{n+1} itx^{2n})_{\vert_{t=0}}\,, \\ 
         &\lim_{N\to\infty}\frac{1}{2N}\cov{Q^{[0]}}{Q^{[0]}} = -\partial^2_{\beta} \cF_{\textrm{Volt}}(\beta,V)\,, 
    \end{split}
    \end{equation}

    where the expected value is taken with respect to the Generalized Gibbs ensemble  $\di\mu_{\Volt}(\ba)$ \eqref{eq:VolterraGibbs}.
    \end{corollary}

    We notice that the result in \cite{Mazzuca2024CLT} is not stated in this way, but this form is more suitable for our analysis.

        \subsection{Currents}

        To continue our analysis, we have to compute the average of the currents. This is usually a difficult task since we do not have a clear connection between the currents averages and some matrix model or the Gibbs ensemble, as in the case of the local conserved fields. Surprisingly, in this case,  as it happened for the Toda lattice , we can compute explicitly these quantities by applying the same idea used in \cite{Spohn2020}, and formalized in \cite{Mazzuca2024CLT}. Specifically, we are able to prove the following:
        
         \begin{lemma}
    \label{lem:mean_currents}
            Consider the Volterra lattice \eqref{Volterra} endowed with the GGE \eqref{eq:VolterraGibbs}, and define the currents $J^{[n]}$ as in \eqref{eq:currents}, then for all fixed $n\in\N$

            \begin{equation}
            	\label{eq:local_current}
                \lim_{N\to\infty} \frac{1}{2N}\meanvalue{J^{[n]}}{1} = -\frac{1}{2}\int_0^\beta \partial_{t_1}\partial_{t_2} \cF_{\textrm{Volt}}(y, V + it_1x^2 + (-1)^{n+1} it_2x^{2n})\di y\,,
            \end{equation}
            where $\cF_{\textrm{Volt}}(\beta,V)$ is the free energy \eqref{eq:free_energy_volterra}.
        \end{lemma}

{\color{black} Since the proof of the previous lemma is rather technical, we defer it to the appendix \ref{app:tech_res}}.
     \subsection{Density of states}
     
     Another fundamental quantity to compute the  linearized correlation functions is the so-called \textit{density of states} of the matrix $L$. We recall that it is defined as the weak limit of the \textit{empirical spectral measures}, i.e. as the probability measure $\di \sigma_{\beta,V}(x)$ such that for any bounded and continuous function $f$ 
     
     \begin{equation}
     	\label{eq:def_DOS}
     	\lim_{N\to \infty}\frac{1}{N}\int_{\R_+}f(x) \sum_{j=1}^{N} \delta_{w_j}(x)  = \int_{\R_+} f(x)\di \sigma_{\beta,V}(x)\,,
     \end{equation}
     where $\pm iw_j$ are the eigenvalues of $L$ \eqref{LaxVolterra}, and we assume that the $w_j$ are positive and in decreasing order. We notice that since the matrix $L$ \eqref{LaxVolterra} is anti symmetric the eigenvalues are purely imaginary number and they come in pairs, meaning that if $iw_j$ is an eigenvalue then $-iw_j$ is also an eigenvalue.

%     
%     For this reason it will be convenient to consider also the density of the square of the eigenvalues, meaning that we are also interested in the distribution of $\lambda_j = w_j^2,\, j=1,\ldots, N$. The following proposition is an immediate consequence of the previous definition:
%     
%     \begin{proposition}
%     	Assume that the density of states of the matrix $L$ \eqref{LaxVolterra} is $d\wt \nu_{\beta,V}(x)$. Define $\lambda_j = w_j^2,\, j=1,\ldots, N$, where $-iw_j, \, j=1, \ldots, 2N$ are the eigenvalues of $L$, and assume that the density of states of $\lambda_j$ is $\di\nu_{\beta,V}(x)$, then for any bounded and continuous function the following holds true
%     	
%     	\begin{equation}
%     		\label{eq:equivalent_measure_gibbs}
%     \int_{\R} f(x)\di\wt {\nu}_{\beta,V}(x) = 	\int_{\R}f(x) \vert x \vert \di\nu_{\beta,V}(x^2) = \frac{1}{2}\int_{\R_+} \left(f(\sqrt{x}) + f(-\sqrt{x}) \right) \di\nu_{\beta,V}(x)\,.
%     	\end{equation}
%     \end{proposition}
%     
%     For later usage, we define
%     
%     \begin{equation}
%     	\rad{f} = \frac{1}{2}\left(f(\sqrt{x}) + f(-\sqrt{x})\right)\,,
%     \end{equation}
%     and we notice that for $\rad{\cdot}$ is a linear operator and for monomials $\rad{w^{2n}} = w^n$ and $\rad{w^{2n+1}} =0$.
%     
     Furthermore, in view of \eqref{eq:def_DOS} and Corollary \ref{cor:mean_cor_gibbs}, we deduce that
     
     \begin{equation}
     	\label{eq:mom_rel}
     	\lim_{N\to\infty}\frac{1}{2N}\meanvalue{Q^{[n]}}{1} = -i\partial_t \cF_{\textrm{Volt}}(\beta,V+(-1)^{n+1}itx^{2n})_{\vert_{t=0}} = \int_{\R_+} w^{2n}\di \sigma_{\beta,V}\,, \quad \forall\, n\in\N\,, n>0\,.
     \end{equation}
     {\color{black}
     The main ingredient of our derivations is to obtain an explicit characterization of the density of states of the Volterra lattice in terms of the one of the \textit{Antisymmetric Gaussian $\beta$ ensemble}. Given this characterization, we are able to explicitly compute the susceptibility matrix $C$ and the current-field matrix $B$, and then the GHD approximation of the correlation functions.}
      
        \section{Antisymmetric Gaussian $\beta$ ensemble in the high temperature regime}
        	\label{sec:antisym}
	The Antisymmetric $\beta$ ensemble is a random matrix ensemble introduced by Dumitriu and Forrester in \cite{Dumitriu_Forrester}; it has the following matrix representation
	
	\begin{equation}
	\label{eq:antiguassian_matrix}
	    Q = \begin{pmatrix}
	    0& y_1 & & &  \\
			-y_1 & 0 & y_2 & &\\
			& \ddots & \ddots & \ddots &\\
			&& \ddots & \ddots &  y_{2N-1}  \\
		 &&&  -y_{2N-1} &0
	    \end{pmatrix}\,,
	\end{equation}
	and the entries of the matrix $Q$ are distributed according to
	
\begin{equation}
\label{eq:antiguassian}
    \di \mu_{AG} = \frac{1}{Z^{AG}_N(\wt \beta,V)}\prod_{j=1}^{2N-1}y_j^{\wt \beta(2N-j)/2 - 1}\mathbbm{1}_{y_j > 0}\exp(\trace{V(Q(\by))})\di \by\,,
\end{equation}
here $V(x)$ can be any function that makes \eqref{eq:antiguassian} normalizable, but for our purpose we  consider  $V(x)$ polynomial of the form $V(x) = (-1)^{\ell+1} c_{\ell}x^{2\ell} + \sum_{j=1}^{\ell-1} c_j x^{2j}$, $c_{\ell} > 0$.
For $V(x) = x^2/2$, it is possible to explicitly compute the partition function $Z_N^{AG}(\wt \beta,x^2/2)$ as

\begin{equation}
	Z_N^{AG}(\wt \beta,x^2/2) = 2^{-2N}\prod_{j=1}^{2N}\Gamma\left(\frac{\wt \beta(2N-j)}{4}\right)
\end{equation}

We are interested in the high-temperature regime for this model, so we set $\wt \beta =\frac{\beta}{N}$, and we rewrite the previous density as
\begin{equation}
\label{eq:antiguassian_ht}
    \di \mu_{AG} = \frac{1}{Z^{AG}_N\left(\frac{\beta}{N},V\right)}\prod_{j=1}^{2N-1}y_j^{\beta\left(1-\frac{j}{2N}\right) - 1}\exp(\trace{V(Q(\by))})\di \by\,\, \quad y_j \geq 0
\end{equation}
This regime was introduced in \cite{Mazzuca2021}, where the authors computed the density of states for this model in the case $V(x) = x^2/2$.
In this particular regime, the partition function  $Z^{AG}_N\left(\frac{\beta}{N},x^2/2\right)$ reads

\begin{equation}
\label{eq:partition_easy_antisym}
	Z^{AG}_N\left(\frac{\beta}{N},\frac{x^2}{2}\right) = 2^{-2N}\prod_{j=1}^{2N} \Gamma\left(\frac{ \beta\left(1 -\frac{j}{2N}\right)}{2}\right)\,.
\end{equation}
\begin{theorem}
	Consider the anti-symmetric $\beta$ ensemble in the high temperature regime \eqref{eq:antiguassian_ht} with potential $V(x) = x^2/2$. The the density of states $\rho_{\beta,V}(y)$ reads
	
\begin{equation}\label{uLdplus1}
	\rho_{\beta,V}(y) = 
	{1 \over \Gamma\left(\frac{\beta}{2} + 1\right) \Gamma\left(\frac{\beta}{2}\right)}
	{|y| \over | W_{\frac{1-\beta}{2},0}(-y^2)|^2},
\end{equation}
 where $W_{\kappa,\mu}$ is the Whittaker function \cite[13.14]{dlmf}.
\end{theorem}

The relation between this model and the Volterra lattice was underlined in \cite{Grava2023,Mazzuca2024CLT}. Specifically, defining the free energy for this model as

\begin{equation}
\label{eq:free_energy_antisym}
    \fF_{\textrm{AG}}(\beta,V) = -\lim_{N\to\infty} \frac{1}{2N} Z_N^{AG}\left(\frac{\beta}{N},V\right)\,,
\end{equation}
from \cite[Remark 2.16]{Mazzuca2024CLT} we deduce the following Proposition

\begin{proposition}
   \label{cor:mean_cor_anti}
    Consider $Q_j^{[n]}$ \eqref{eq:locally_conserved_fields}, the Generalized Gibbs ensemble $\di\mu_{\Volt}(\ba)$ \eqref{eq:VolterraGibbs}, the free energy $\cF_{\textrm{Volt}}(\beta,V)$ \eqref{eq:free_energy_volterra}, and the free energy $\fF_{\textrm{AG}}(\beta,V)$ \eqref{eq:free_energy_antisym}, then for any fixed $n,m\in \N$  the following holds true

    \begin{equation}
    \begin{split}
    &\partial_\beta (\beta \fF_{\textrm{AG}}(\beta,V)) = \cF_{\textrm{Volt}}(\beta,V)\\
         &\lim_{N\to\infty}\frac{1}{2N}\meanvalue{Q^{[n]}}{1} = -i\partial_t\partial_\beta  (\beta \fF_{\textrm{AG}}(\beta,V+(-1)^{n+1}itx^{2n}))_{\vert_{t=0}}\,, \\
         &\lim_{N\to\infty}\frac{1}{2N}\meanvalue{Q^{[0]}}{1} = -\partial^2_\beta (\beta \fF_{\textrm{AG}}(\beta,V)) \\ 
         &\lim_{N\to\infty}\frac{1}{2N}\cov{Q^{[n]}}{Q^{[m]}} = \partial_{t_1}\partial_{t_2}\partial_\beta (\beta \fF_{\textrm{AG}}(\beta,V+(-1)^{n+1}it_1x^{2n}+(-1)^{m+1}it_2 x^{2m}))_{\vert_{t_1=t_2=0}}\,, \\
         &\lim_{N\to\infty}\frac{1}{2N}\cov{Q^{[n]}}{Q^{[0]}} = -i\partial_{t}\partial^2_{\beta}(\beta \fF_{\textrm{AG}}(\beta,V+(-1)^{n+1}itx^{2n}))_{\vert_{t=0}}\,, \\ 
         &\lim_{N\to\infty}\frac{1}{2N}\cov{Q^{[0]}}{Q^{[0]}} = -\partial^3_{\beta}(\beta \fF_{\textrm{AG}}(\beta,V))\,, 
    \end{split}
    \end{equation}

    where the expected value is taken with respect to the Generalized Gibbs ensemble  $\di\mu_{\Volt}(\ba)$ \eqref{eq:VolterraGibbs}.
    
\end{proposition}

{ \color{black}

\begin{proof}
In view of Corollary \ref{cor:mean_cor_gibbs}, one has just to show that

\[ \partial_\beta (\beta \fF_{\textrm{AG}}(\beta,V)) = \cF_{\textrm{Volt}}(\beta,V)\,. \]

For $V(x) = \frac{x^2}{2}$, the previous equality can be deduced from the explicit expression of the  partition functions \eqref{eq:partition_easy_volterra}-\eqref{eq:partition_easy_antisym}. The case for polynomial $V(x)$ is a corollary of \cite[Theorem 1.5]{Mazzuca2024CLT}.
\end{proof}
}
\subsection{Density of states}
The density of states $\rho_{\beta,V}$ for the Anti-symmetric $\beta$ ensemble can be computed explicitly when the potential $V(x) = x^2/2$. For general polynomial potential, we can characterize the density of states for this ensemble using a \textit{Large Deviation principle} (LDP) \cite{LDPbook}. This is not surprising, indeed for all the $\beta$ ensembles this is true, see \cite{forrester}. 
In our case, the LDP is a corollary of \cite[Theorem 1.2]{Zelada2019} in combination with the result of Dumitriu--Forrester \cite{Dumitriu_Forrester}, who were able to compute explicitly the joint eigenvalue density of the Anti-symmetric Gaussian $\beta$ ensemble

\begin{theorem}
	\label{thm:DF_thm}
	Consider the anti-symmetric $\beta$ ensemble \eqref{eq:antiguassian}, and let $iw_j\, j=1,\ldots, N$ be the first $N$ ordered eigenvalues $w_1\geq w_2\geq \ldots\geq w_N>0$ of the matrix $Q$ \eqref{eq:antiguassian_matrix} endowed with the distribution \eqref{eq:antiguassian}, where the potential $V(x)$ is such that
	\begin{equation}
		\label{eq:growth_condition}
		 \lim_{|x|\to\infty} \frac{|V(x)|}{\ln(|x|)} = +\infty\,,
	\end{equation}
	
	then the probability density function (PDF) for $w_1,\ldots,w_N$ is given by
	
	\begin{equation}
		\frac{1}{\mathfrak{C}_{N,\wt \beta,V}} \prod_{j=1}^{N}w_j^{\wt\beta/2 -1}e^{\sum_{j=1}^N V(w_j) + V(-w_j)}\prod_{1\leq j < i \leq N} \left(w_j^2 - w_i^2 \right)^{\wt \beta}d \bw\,.
	\end{equation}

For $V(x) = x^2/2$ , one can explicitly evaluate $\mathfrak{C}_{N,\beta/N,x^2/2}$ as

\begin{equation}
	\mathfrak{C}_{N,\wt \beta,x^2/2}  = \frac{1}{N!}\prod_{j=1}^{N} \frac{\Gamma\left(1+ \frac{j\wt\beta}{2}\right) \Gamma\left(\frac{(2j-1)\wt\beta}{4}\right)}{2\Gamma\left(1+ \frac{\wt \beta}{2}\right)} 
\end{equation}

Furthermore, let $q_j$ $j=1,\ldots, N$ be the (positive) first components of the independent eigenvector corresponding to $iw_j$. Then, the vector $(q_1,\ldots, q_N)$ has a Dirichlet distribution $D_N[(\wt\beta/2)^N]$ (here $(\wt\beta/2)^N$ denotes $\wt \beta/2$ repeated $N$ times).

\end{theorem} 
We notice that the previous theorem is stated in \cite{Dumitriu_Forrester} just for the case $V(x) = x^2/2$, but it is straighforeward to generalize it for potential $V(x)$ satisfying condition \eqref{eq:growth_condition}.
One of the key steps of the proof of Dumitriu and Forrester is the explicit computation of the Jacobian  of $\Phi \,:\, \by \to (\bw,\bq)$. Specifically, they proved the following:

\begin{equation}
	\label{eq:jacobian}
	\prod_{j=1}^{2N}2y_j^{\wt \beta j/2 - 1} d\by = \left(c_q^{\wt\beta} \prod_{j=1}^{N}q_j^{\wt \beta -1}d \bq\right)\left(\zeta_N(\wt\beta)\prod_{j=1}^{N}w_j^{\wt\beta/2 -1}\prod_{1\leq j < i \leq N} \left(w_j^2 - w_i^2 \right)^{\wt \beta}d  \bw \right)
\end{equation}

Where, $c_q^{\wt \beta} = 2^{N-1}\Gamma\left(\frac{1}{2}\wt \beta N\right) \Gamma\left(\frac{1}{2}\wt \beta \right)^{-N}$ which normalize to $1$ the first term, and $\zeta_N(\wt \beta)$ is given by

\begin{equation}
	\label{eq:zeta}
	\zeta_N(\wt \beta) = \mathfrak{C}_{N,\wt \beta,x^2/2}^{-1}\prod_{j=1}^{2N}\Gamma\left(\frac{\wt \beta j}{2}\right)\frac{1}{N!}\,.
\end{equation}
 We are interested in the high temperature regime for this ensemble, {\color{black} i.e. the regime where $\wt \beta = \beta/N$. From the previous discussion,} we deduce the following

\begin{corollary}
	In the same hypotheses as before, let $\wt \beta = \beta/N$,
	then the probability density function (PDF) for $w_1,\ldots,w_N\in\R_+$ is given by
	
	\begin{equation}
		\frac{1}{\mathfrak{C}_{N,\beta/N,V}} \prod_{j=1}^{N}w_j^{\frac{\beta}{2N} -1}e^{\sum_{j=1}^N V(iw_j) + V(-iw_j)}\prod_{1\leq j < i \leq N} \left(w_j^2 -w_i^2 \right)^{\frac{\beta}{N}}d \bw\,.
	\end{equation}

Furthermore, for $V(x) = x^2/2$ , one can explicitly evaluate $\mathfrak{C}_{N,\beta/N,x^2/2}$ as

\begin{equation}
	\label{eq:scaled_constant}
	\mathfrak{C}_{N,\beta/N,x^2/2}  = \frac{1}{N!}\prod_{j=1}^{N} \frac{\Gamma\left(1+\frac{j\beta}{2N}\right) \Gamma\left(\frac{(2j-1)\beta}{4N}\right)}{2\Gamma\left(1+ \frac{\beta}{2N}\right)} 
\end{equation}

\end{corollary} 

Using the previous result  combined with \cite[Theorem 1.1]{Zelada2019} we deduce the following

\begin{theorem}
	\label{thm:LDP_antisym}
	Consider the functional $\cF_{\textrm{AG}}(\beta,V)[\rho]$ defined as 
	\begin{equation}
		\label{eq:functional}
		\begin{split}
					\cF_{\textrm{AG}}(\beta,V)[\rho] &		= -\frac{\beta}{2} \int\int_{\R_+^2}\ln(\vert x^2-y^2 \vert)  \rho(x) \rho(y)\di x \di y - \int_{\R_+}(V(ix)  +V(-ix)  + \ln(|x|)) \rho(x)\di x \\ & +\int_{\R_+}\ln( \rho(|x|)) \rho(|x|)\di x 
		\end{split}
	\end{equation}
	here $\rho(x)$ is an absolutely continuous measure with respect to the Lebesgue one, has support on the positive real line. The previous functional has a unique minimizer  $\rho_{\beta,V}(x)$, which is absolutely continuous with respect to the Lebesgue measure. In particular, $\rho_{\beta,V}(x)$ is the density of states of the Anti-symmetric $\beta$ ensemble in the high temperature regime. Furthermore, 
	
	\begin{equation}
		\fF_{\textrm{AG}}(\beta,V)  = \frac{1}{2}\cF_{\textrm{AG}}(\beta,V)[\rho_{\beta,V}]  + \frac{1}{2}\int_0^1 \ln\left(\frac{\beta}{2}x\right) dx - \frac{\ln(2)}{2}\,.
	\end{equation}
%	 Furthermore, the following holds for any continuous and bounded function $f$
%	
%     	\begin{equation}
%	\int_{\R} f(x)d\wt \rho_{\beta}^V(x) = 	\int_{\R_+}\rad{f} d\rho_{\beta}^V(x) = \int_{\R}\vert x \vert f(x) d\rho_{\beta}^V(x^2) \,.
%\end{equation}
%	
\end{theorem}
{\color{black} Since the proof of the previous Theorem is rather technical, we defer it to the appendix \ref{app:tech_res}}

\begin{remark}
	We notice that, following the same procedure as in \cite{Guionnet2022,Mazzuca2023}, it would be possible to obtain a LDP also for the Volterra lattice, and generalize Corollary \ref{cor:mean_cor_anti} for a general potential satisfying \eqref{eq:growth_condition}.
\end{remark}

\section{On the way to the matrix $C$ and $B$}
\label{sec:computationCB}
{\color{black} In this section, we prove our main result, i.e. Theorem \ref{thm:main}. Therefore, we show how to compute the susceptibility matrix $C$ and the charge-current correlation matrix $B$.

The main idea is to enforce Proposition \ref{cor:mean_cor_anti} and Lemma \ref{lem:mean_currents} to obtain an explicit expression of $C,B$ in terms of the free energy $\fF_{\text{AG}}$ of the anti-symmetric Gaussian $\beta$ ensemble at high temperature, and then use Theorem \ref{thm:LDP_antisym} to rewrite them in terms of the density of states $\rho_{\beta,V}$ of the anti-symmetric ensemble. Finally, we use the relation between $\rho_{\beta,V}$ and $\sigma_{\beta,V}$, the density of states of the Volterra lattice, to conclude the proof.}

From Corollary \ref{cor:mean_cor_anti}, we know that 

\begin{equation}
	\partial_\beta (\beta \fF_{\textrm{AG}}(\beta,V)) = \cF_{\textrm{Volt}}(\beta,V)\,,
\end{equation}
which combined with Theorem \ref{thm:LDP_antisym} gives

\begin{equation}
	\cF_{\textrm{Volt}}(\beta,V) = \partial_\beta\left(\frac{\beta}{2} \cF_{\textrm{AG}}(\beta,V)[\rho_{\beta,V}]\right) +\frac{\ln(\beta)}{2} - \ln(2)\,.
\end{equation}
For the following computations, it is more convenient to absorb $\beta$ into the measure $\rho$ by setting $\varrho = \beta \rho$, in this way we get the modified functional from $\beta \cF_{\textrm{AG}}(\beta,V)[\beta^{-1}\varrho] = \cF[\varrho] - \beta\ln\left(\beta\right)$, where

\begin{equation}
	\label{eq:var_no_beta}
	\cF[\varrho] = -\frac{1}{2}\int\int_{\R^2_+}\ln(\vert x^2-y^2 \vert) \varrho(x)\varrho(y)\di x \di y - \int_{\R_+}(V(ix)  + V(-ix) - \ln(|x|)) \varrho(x)\di x +\int_{\R_+}\ln(\varrho(x))\varrho(x)\di x \,,
\end{equation}
which has to be minimized under the condition that

\begin{equation}
	\varrho \geq 0\,,\qquad \int_{\R_+} \varrho(x)\di x = \beta\,.
\end{equation}

We define the unique minimizer $\varrho^\star$; {\color{black} we notice that  $\varrho^\star = \beta \rho_{\beta,V}$}. Then

\begin{equation}
	\cF_{\textrm{Volt}}(\beta,V) = \frac{1}{2}\partial_\beta \cF[\wt \varrho^\star]  - \ln(2) - \frac{1}{2}  \,.
\end{equation}
The minimizer $\varrho^\star$ is characterized by the Euler-Lagrange equation

\begin{equation}
	\label{eq:euler_lagrange}
-\int_{\R}\ln(\vert x^2-y^2 \vert) \varrho^\star(y)\di y  -(V(ix) + V(-ix))  +\ln(|x|) + \ln(\varrho^\star) + 1 -  \mu(\beta,V) = 0\,,
\end{equation}
where $\mu(\beta,V)$ is a function depending on $\beta,V$, {\color{black} but not on $x$}.

To obtain the free energy of the Volterra lattice, we differentiate the functional as

\begin{equation}
	\begin{split}
			\partial_\beta \cF[\varrho^\star] & =  -\int_{\R^2}\ln(\vert x^2-y^2 \vert)\partial_\beta\varrho^\star(x) \varrho^\star(y)\di x\di y  - \int_{\R}(V(ix)  + V(-ix) - \ln(|x|))\partial_\beta \varrho^\star(x)\di x \\ &+ \int_{\R}\ln( \varrho^\star)\partial_\beta \varrho^\star(x) + \int_{\R}\partial_\beta \varrho^\star(x) \\ & = -\int_{\R^2}\ln(\vert x^2-y^2 \vert)\partial_\beta\varrho^\star(x) \varrho^\star(y)\di x\di y  - \int_{\R}(V(ix)  + V(-ix) - \ln(|x|))\partial_\beta \varrho^\star(x)\di x \\ &+ \int_{\R}\ln( \varrho^\star)\partial_\beta \varrho^\star(x) + 1  \,.
	\end{split}
\end{equation}
By testing \eqref{eq:euler_lagrange} against $\partial_{\beta}\varrho^\star$ we deduce that

\begin{equation}
	\partial_\beta \cF[\varrho^\star] =  \mu(\beta,V)\,,
\end{equation}
which implies that
\begin{equation}
	\label{eq:free_energy_rel}
	\cF_{\textrm{Volt}}(\beta,V) = \frac{\mu(\beta,V)}{2}  - \ln(2) - \frac{1}{2}\,.
\end{equation}

Consider now the following chain of equality

\begin{equation}
	\partial_\mu \varrho^\star = \partial_\mu (\beta \rho_{\beta,V}) = \partial_{\beta}(\beta \rho_{\beta,V}) (\partial_{\beta}\mu)^{-1} = \sigma_{\beta,V} \kappa^{-1}\,,
\end{equation}
where we defined 

\[\kappa= 2\partial_\beta \cF_{\textrm{Volt}}(\beta,V) =-\meanvalue{\ln(a_1)}{1} .\]

 Following Spohn, we define a new measure $\sigma= \sigma_{\beta,V}\kappa^{-1}$, and we notice that $\la \sigma \ra = \kappa^{-1}$.

The measure $\sigma,\varrho^\star$ play a crucial role in the computation of the matrices $B,C$. To simplify the notation we drop the upper index $\star$ from $\varrho^\star$ . Before proceeding with the computation of such matrices, we have to introduce the following operator

\begin{equation}
	\label{eq:def_T}
	T\psi(w) = \int_{\R_+}\ln(\vert w^2-z^2\vert )\psi(z)\di z\,\quad w\in \R\,,
\end{equation}
{\color{black} which is related to the scattering shift of the Volterra lattice.} Using this operator, we can introduce the \textit{dressing of a function }$\psi$

\begin{equation}
	\label{eq:dressing}
	\dr{\psi} =  \psi + T\varrho\dr{\psi}\,, \quad \dr{\psi} = (1-T\varrho)^{-1}\psi\,, 
\end{equation}
here $\varrho$ is just a multiplicative operator. We notice that the dressing of any real function according to \eqref{eq:dressing} is \textit{even}, {\color{black} therefore if the function is differentiable at $0$, then the dressed function has a stationary point at $0$}. Furthermore, by differentiating the Euler-Lagrange equation with respect to $\mu$ we deduce the following chain of equality
\begin{equation}
	\label{eq:useful}
	\sigma = \left(1-\varrho T \right)^{-1} \varrho = \varrho\left(1- T \varrho \right)^{-1}[1] = \varrho\dr{1} \,,
\end{equation}
where we used the fact that $\left(1-\rho T \right)^{-1} \rho = \rho\left(1- T \rho \right)^{-1}[1]$ for any measure $\rho$.

Using this notation, we can express the moments of the Volterra lattice as

\begin{equation}
	q_n = \meanvalue{Q_0^{[n]}}{1}  = \kappa \jap{ \sigma w^{2n}}\,,
\end{equation}

where for any function $f$ we defined $\jap{f} = \int_{\R_+} f(w) \di w$. 

The following Proposition contains several properties of the dressing operator and the measure $\varrho$ that we use to compute the matrices $B,C$. {\color{black} Since the proof is rather technical, we defer it to the appendix \ref{app:tech_res}.}

\begin{proposition}
\label{prop:dress_prop}
	Consider the measure $\varrho$ defined as the unique minizier of \eqref{eq:var_no_beta}, the operator $T$ defined in \eqref{eq:def_T} and the dressing operator \eqref{eq:dressing}. Then the following holds true
	
	\begin{enumerate}
		\item for any function $f$
		\begin{equation}
			\label{eq:prop1}
			(1-\varrho T)^{-1}[\varrho f] = \varrho(1- T\varrho)^{-1}[f]\,,
		\end{equation}
		
		\item  for any function $f,g$ 
		\begin{equation}
			\label{eq:prop2}
			\jap{(1-\varrho T )^{-1}[f] g} = \jap{f \dr{g}}\,,
		\end{equation}
		\item for any variable $\odot$
		
		\begin{equation}
			\label{eq:prop3}
			\partial_\odot \partial_\mu \varrho = (1-\varrho T)^{-1} \partial_\odot \varrho (1- T\varrho)^{-1}[1]
		\end{equation}
		
		\item for any function $f$
		
		\begin{equation}
			\label{eq:prop4}
			\partial_\mu\jap{\sigma f} = \jap{\sigma \dr{1}\dr{f}}\,.
		\end{equation}
		
		\item Consider the perturbed potential $V(x) \to V(x) +(-1)^{n+1}it_nx^{2n}$ then 
		
		 \begin{equation}
		 	\label{eq:moments_mu}
			\partial_{t_n} \mu(\beta, V(x) + (-1)^{n+1} it_n x^{2n})_{\vert_{t_n=0}} = 2iq_n\,.
		\end{equation} 
		
		\item For any function $\psi$
		
		\begin{equation}
			\label{eq:prop6}
			\partial_{t_n}\left( (1-T\varrho(V+ (-1)^{n+1}it_n\lambda^{2n}))^{-1}(\psi)\right)_{\vert_{t_n=0}} = (1-T\varrho)^{-1}\left(T\partial_{t_n}\varrho(V+ (-1)^{n+1}it_n\lambda^{2n})_{\vert_{t_n=0}}\dr{\psi}\right)\,.
		\end{equation}
	\end{enumerate}
\end{proposition}

{\color{black} Given the previous proposition, we can compute the matrices $C,B$.}
\subsection{The matrix $C$}

The aim of this section is to compute the correlation matrix $C$ defined as

\begin{equation}
	C_{m,n} = \sum_{j=1}^{2N} \cov{Q_j^{[n]}}{Q_0^{[m]}}\,.
\end{equation}

We start with $C_{0,0}$

\begin{equation}
	\label{eq:c00}
	\begin{split}
			C_{0,0} \stackrel{\text{Cor} \ref{cor:mean_cor_gibbs}}{=}&  -\partial_{\beta}^2\cF_{\textrm{Volt}}(\beta,V) \stackrel{\eqref{eq:free_energy_rel}}{=}  - \frac{1}{2}\partial_{\beta}^2\mu(\beta,V) = -\frac{1}{2}\partial_{\beta} \kappa \\&=-\frac{1}{2}\partial_{\beta} \frac{1}{\jap{\sigma}} = \frac{\kappa ^2}{2} \partial_{\beta} \jap{\sigma} = \frac{\kappa ^2}{2}\partial_{\beta} \mu \partial_{\mu} \jap{\sigma} \stackrel{\eqref{eq:prop4}}{=} \frac{\kappa^3}{2}\jap{\sigma(\dr{1})^2}
	\end{split}
\end{equation}

Next we consider $C_{0,n} = C_{n,0} $, from Corollary \ref{cor:mean_cor_gibbs} we deduce that

\begin{equation}
	\begin{split}
		C_{0,n} &= \partial_{\beta} \jap{\kappa \sigma w^{2n}} = (\partial_{\beta}\kappa)\jap{\sigma w^{2n}} + \kappa \partial_{\beta}\jap{\sigma w^{2n}}\\
		& \stackrel{\eqref{eq:c00} - \eqref{eq:prop4}}{=} -\kappa^3\jap{\sigma(\dr{1})^2}\jap{\sigma w^{2n}} + \kappa^2\jap{\sigma \dr{1}\dr{w^{2n}}} \\
		& = -\kappa^2 \jap{\sigma(\dr{1})^2}q_n  +\kappa^2\jap{\sigma \dr{1}\dr{w^{2n}}} =  \kappa^2 \jap{\sigma\dr{1}\left(\dr{w^{2n}} -q_n\dr{1}\right)}
	\end{split}
\end{equation}

Finally, we have to compute $C_{n,m} = C_{m,n}$, from Corollary \ref{cor:mean_cor_gibbs} we deduce that

\begin{equation}
	C_{n,m} = i\partial_{t_n}\jap{\kappa \sigma(V-it_nw^{2n}) w^{2m}}_{\vert_{t_n=0}} = i\frac{\partial_{t_n}\kappa_{\vert_{t_n=0}}}{\kappa}q_m + i\jap{\kappa w^{2m}\partial_{t_n}\sigma(V+ (-1)^{n+1} it_n w^{2n})_{\vert_{t_n=0}}} \,.
\end{equation}

We have $i\partial_{t_n}\kappa_{\vert_{t_n=0}} = -2C_{n,0}$. Regarding the second derivative we use the free energy
\begin{equation}
	\begin{split}
		\cF[\varrho] & = -\frac{1}{2}\int\int_{\R_+^2}\ln(\vert w^2-\lambda^2 \vert) \varrho(w)\varrho(\lambda)\di w \di \lambda +  \int_{\R_+}(-V(iw)  -V(-iw)+ 2it_nw^{2n})\varrho(w)  \\ & +\int_{\R_+}\ln(|w|)) \varrho(w)\di w  +\int_{\R_+}\ln(\varrho(w))) \varrho(w)\di w \,,
	\end{split}
\end{equation}
so that 
\begin{equation}
\begin{split}
	\dfrac{\delta \mathcal{F}}{\delta \varrho}&= -\int_{\R}\ln(\vert w^2-\lambda^2 \vert)\varrho(\lambda)\di \lambda  -(V(iw)+V(-iw)) + 2it_nw^{2n} +\ln(|w|) + \ln(\varrho) + 1 -  \mu(\beta,V)  \\
	& =  -T\varrho(w)  -(V(w)+V(-|w|)) + 2it_nw^{2n} +\ln(|w|) + \ln(\varrho) + 1 - \mu(\beta,V) = 0\,.
\end{split}
\end{equation}

Taking the derivative with respect to  $t_n$ we obtain  the equations
\begin{equation}
	\label{var_rhomu}
	2iw^{2n}-T\left(\dfrac{\partial}{\partial t_n}\varrho(w)\right)+\dfrac{\dfrac{\partial}{\partial t_n}\varrho(w)}{\varrho(w)}-\dfrac{\partial}{\partial t_n}\mu(\beta,V)=0\,,
\end{equation}
so that
\begin{equation}
	\label{eq:time_der}
	(1-\varrho T)\left(\dfrac{\partial}{\partial t_n}\varrho(w)\right)=\left(\dfrac{\partial}{\partial t_n}\mu- 2iw^{2n}\right)\varrho(w)\,.
\end{equation}
Now taking the derivative with respect to $\mu$ we obtain 
\[
(1-\varrho T)\left(\dfrac{\partial}{\partial t_n}\sigma(w)\right)-\sigma T\left(\dfrac{\partial}{\partial t_n}\varrho(w)\right)=\left(\dfrac{\partial}{\partial t_n}\mu-2i{w^{2n}}\right)\sigma(w)\,.
\]
Applying \eqref{eq:time_der}, we deduce
\[
(1-\varrho T)\left(\dfrac{\partial}{\partial t_n}\sigma(w)\right)=\frac{\sigma(w)}{\varrho(w)}(1-\varrho T)^{-1}\left(2iq_n-2i{w^{2n}}\right)\varrho(w)\,,
\]
therefore 
\begin{align}
	\langle\partial_{t_n}\sigma(V+t_nw^{2n}){w^{2m}}\rangle&=\langle {w^{2m}}(1-\varrho T)^{-1} \left( \frac{\sigma}{\varrho}(1-	\varrho T)^{-1}\left(2iq_n-2i{w^{2n}}\right)\varrho(w) \right) \rangle\\
	&=-2i\langle \sigma\dr{{w^{2m}}}\left(\dr{{w^{2n}}} - q_n[1]^{dr}\right)\rangle\,.
\end{align}

So, we show that

\begin{equation}
	C_{n,m} = 2\kappa \jap{\sigma\left(\dr{{w^{2m}}} - q_m[1]^{dr}\right) \left(\dr{{w^{2n}}} - q_n[1]^{dr}\right)}\,.
\end{equation}

\subsection{The matrix $B$}

The matrix $B$ is the matrix of static covariance between the conserved fields and currents, which is defined as

\begin{equation}
	B_{n,m} =\lim_{N\to \infty} \frac{1}{2N}\cov{Q^{[n]}}{J^{[m]}}\,.
\end{equation}

A priori, this matrix is not symmetric, but, as we show in this section, it is.

First, we start by computing $B_{n,0} = \lim_{N\to \infty}(2N)^{-1}\cov{Q^{[n]}}{J^{[0]}}$. From Remark \ref{rem:trick_rem}, we deduce that

\begin{equation}
	B_{n,0} = -\frac{1}{2}\lim_{N\to \infty}(2N)^{-1}\cov{Q^{[n]}}{Q^{[1]}} =-\frac{1}{2}C_{n,1}\,.
\end{equation}

To compute the remaining part of the matrix $B$, we need to express $\meanvalue{J_0^{[n]}}{1}$ \eqref{eq:local_current} using our new notation

\begin{equation}
	\begin{split}
			\meanvalue{J_0^{[n]}}{1} & = -\frac{1}{2}\int_0^\beta \partial^2_{t_1,t_2} \cF_{\textrm{Volt}}(y, V-it_1x^2 -it_2x^{2n})\di y \\ & 
			= -\frac{1}{4}\beta \partial_{t_1} \partial_{t_2} \fF_{\textrm{AG}}(\beta, V-it_1x^2 -it_2x^{2n}) = -\frac{1}{2}i \partial_{t_n}\jap{\varrho(V-it_nw^{2n}){w^2}} \\
			& \stackrel{\eqref{eq:time_der}}{=}  \jap{\varrho \dr{{w^2}} ({w^{2n}} - q_n)}
	\end{split}\,.
\end{equation}

Defining 

\begin{equation}
	\label{eq:v_eff}
	v_{\textrm{eff}} = \frac{\dr{{w^2}}}{\dr{1}}\,,
\end{equation}
we can recast the previous expression as

\begin{equation}
	\label{eq:my_current}
	\meanvalue{J_0^{[n]}}{1} =  \jap{\sigma v_{\textrm{eff}} ({w^{2n}} - q_n)} =  \jap{\sigma (v_{\textrm{eff}} - q_1) ({w^{2n}} - q_n)}\,.
\end{equation}
%Furthermore, by exchanging the order of derivation, we deduce the following equalities
%
%\begin{equation}
%	\label{eq:Spohn_current}
%	\begin{split}
%		\meanvalue{J_0^{[n]}}{1}& = \jap{\sigma{w^n} (v_{\textrm{eff}} - q_1)} \\ 
%		\jap{\sigma v_{\textrm{eff}}}& = \frac{\jap{\sigma {w^n}q_1}}{q_n}
%	\end{split}\,.
%\end{equation}

We can now compute $B_{n,m}$ as follows

\begin{equation}
	\begin{split}
		B_{n,m} & =- i\partial_{t_n}\meanvalue{J_0^{[m]}}{1} \stackrel{\eqref{eq:my_current}}{=} -i \partial_{t_n}\jap{\varrho (\dr{{w^2}} - q_1\dr{1}) ({w^{2m}} -q_m)}\\
		& \stackrel{\eqref{eq:prop6} - \eqref{eq:time_der}}{=}- 2\jap{\varrho ({w^{2m}} -q_m) \dr{{w^{2n}} - q_n} (\dr{{w^2}} - q_1\dr{1})} \\ & - 4\jap{T\varrho\dr{{w^{2m}} -q_m} \varrho\dr{{w^{2n}} - q_n}(\dr{{w^2}} - q_1\dr{1})}\\
		& \stackrel{\eqref{eq:dressing}}{=}  -2\jap{\varrho \dr{{w^{2m}} -q_m} \dr{{w^{2n}} - q_n} (\dr{{w^2}} - q_1\dr{1})} = - 2\jap{\sigma (v_{\textrm{eff}} - q_1) \dr{{w^{2m}} -q_m} \dr{{w^{2n}} - q_n}}
	\end{split}\,.
\end{equation}

From the previous equation, we deduce that $B$ is actually symmetric.

This concludes the proof of Theorem \ref{thm:main} \qed

\begin{remark}
\label{rem:collision}
	We notice that from our definition of $v_{\textrm{eff}}(w)$, one can deduce the following \textit{Collision rate ansatz} for the effective velocity
	
	\begin{equation}
		v_{\textrm{eff}} = w^2 + T\sigma[\ve] - \ve T\sigma[1]\,.
	\end{equation}
\end{remark}

\begin{proof}
	
	Multiplying the definition of $\ve$ by $\sigma$ we deduce that
	\begin{equation}
		\sigma \ve = (1-\varrho T)^{-1}[\varrho w^2]\,,
	\end{equation} 
	which also reads
	
	\begin{equation}
		(1-\varrho T)[\sigma \ve] = \varrho w^2\,.
	\end{equation}
	From equation \eqref{eq:useful}, we deduce that $\frac{\sigma}{\varrho} = 1+T\sigma[1]$ thus
	
	\begin{equation}
		\label{eq:collision}
		\ve(1+T\sigma[1]) = w^2 + T\sigma[\ve]\,,
	\end{equation}
	rearranging the previous equation we deduce our claim.
\end{proof}

For later computations, the basis of moments that we are considering while computing the matrices $B,C$ is not convenient. For this reason we introduce the space $\mathbb{C}\oplus L^2(\sigma, \{1\}^{\perp})$, where by  $L^2(\sigma, \{1\}^{\perp})$ we denote the space of square integrable function with respect to $\sigma$ such that they are orthogonal to the constant function. Defining the operator $\Xi$ and its adjoint $\Xi^*$ as

\begin{equation}
	\Xi \phi = \dr{\phi - \jap{\kappa \sigma \phi}}\,, \qquad \Xi^*\psi = (1-\varrho T)^{-1}\psi - \kappa\sigma \jap{\dr{1}\psi}\,.
\end{equation}

Using the notation that we have just introduced, we define the following matrix operators

\begin{equation}
	\label{eq:BC_operator}
	\begin{split}
			& \fC = \begin{pmatrix}
			\frac{\kappa^3}{2}\jap{\sigma (\dr{1})^2}  & \kappa\langle \Xi^* \kappa\sigma \dr{1}\vert   \\ 
			\kappa\vert \Xi^* \kappa\sigma \dr{1}\rangle  & 2\Xi^* \kappa\sigma \Xi
		\end{pmatrix} \,, \\ & 
		 \fB = -\frac{1}{\kappa} \begin{pmatrix}
			\frac{\kappa^3}{2} \jap{\sigma (\dr{1})^2(\ve - q_1)} & \kappa\langle \Xi^* \kappa\sigma(\ve - q_1) \dr{1}\vert   \\ 
			\kappa\vert \Xi^* \kappa\sigma(\ve - q_1) \dr{1}\rangle  &2 \Xi^* \kappa\sigma(\ve - q_1) \Xi
		\end{pmatrix}\,.
	\end{split}
\end{equation}
In this notation we can recast the matrices $B,C$ as

\begin{equation}
	C_{0,0} = \fC_{0,0}\,,\quad C_{n,0} = C_{n,0} = \fC_{0,1}[w^n]\,,\quad C_{m,n}= C_{n,m}=\langle w^m ;\fC_{1,1}[w^n] \rangle \,, 
\end{equation}
and analogously for $B$.

\section{Linearized Hydrodynamics}

\label{sec:linearhydro}
In this section, we compute the correlation functions of the Volterra lattice using the theory of Generalized Hydrodynamics. We start by computing the Euler equation for the density. We start from the continuity equation

\begin{equation}
	\label{eq:starting}
	\partial_t Q_j^{[n]} = -\partial_x J_j^{[n]}\,,
\end{equation}
which, by averaging on a GGE with slowly varying parameters, become

\begin{equation}
\label{eq:first_approx}
	\partial_t \jap{Q_j^{[n]}}_{\epsilon} = -\partial_x \jap{J_j^{[n]}}_{\epsilon}\,,
\end{equation}
where by $\jap{\cdot}_{\epsilon}$ we denote the parameter with slow variation, and by $\partial_x\cdot$ the discrete spatial derivative.

After a normalization procedure, which still needs some rigorous justification, the previous equations imply that the density $\sigma$ and the normalization $\kappa$ evolve according to the following system of quasi-linear equations

\begin{equation}
	\label{eq:linearzied_1}
\partial_t \kappa = \partial_x 	q_1\,,\qquad \partial_t(\kappa \sigma) + \partial_x((\ve - q_1)\sigma) = 0\,,
\end{equation}
at Euler scale.

As in the case of Toda lattice, the previous equations can be put in linear form by the following change of coordinates

\begin{equation}
	\varrho = \sigma (1+T\sigma)^{-1}[1]\,,
\end{equation}
in this new variable the equations read \eqref{eq:linearzied_1} 

\begin{equation}
	\kappa \partial_t \varrho + (\ve - q_1)\partial_x \varrho = 0\,,
\end{equation}
the proof is analogous to the one in \cite{Spohn2020}. {\color{black} From the previous expression, we notice that the solution of this system of equations can develop some shock, indeed  if the velocity $\ve(w) - q_1 < 0$ the corresponding wave will encouter some sort of obstacle at $x=0$, since the density is defined just for positive $x$.}
We notice that this behavior is an effect of the linearization procedure. We expect that if we would to consider a more accurate description of the model by making a second order average approximation of equation \eqref{eq:starting}, i.e. adding a damping term in the form of a Drude weight, the evolution would be smooth and the shock disappear. For a more general discussion, we refer to \cite[Chapter 12]{spohn2021hydrodynamicToda} and \cite{DeNardis2018}.

Despite that, we can still apply the theory of GHD (Landau-Lifshitz theory) to describe the correlation functions. For a general introduction see \cite[Chapter 7]{spohn2021hydrodynamicToda}. The structure of the matrices $B,C$ is the same as in \cite{Spohn2020,spohn2021hydrodynamicToda}, thus  following the exact same reasoning, {\color{black} we can guess the linear correlation functions has the following operator form

\begin{equation}	
	\label{eq:sol_linear}
	\tS(x,t)= \begin{pmatrix}
		\frac{\kappa^3}{2}\jap{\sigma \delta(x+t(\ve - q_1)\kappa^{-1})\left((\dr{1})^2 \right)} & \kappa \langle \Xi^*\kappa \sigma \delta(x+t(\ve - q_1)\kappa^{-1})\dr{1}\vert \\
		\kappa \vert  \Xi\kappa \sigma \delta(x+t(\ve - q_1)\kappa^{-1})\dr{1}\rangle & 2\Xi^*\sigma \kappa  \delta(x+t(\ve - q_1)\kappa^{-1}) \Xi 
	\end{pmatrix}\,.
\end{equation}
}
As we already noticed, these equations develops shock. Here this effect is clearer in view of the structure of the effective velocity $\ve(w)$, see Figure \ref{fig:v_eff}. Indeed, by looking at the collision rate ansatz \eqref{eq:collision}, one immediately deduces that the effective velocity is an even function, and it is not singular, thus it is not a one to one transform of $\R_+$ into $\R$. So equations \eqref{eq:sol_linear}  are not continuous for some values of $ \frac{x}{t} =\xi_0$ that we can compute explicitly as

\[ \xi_0 = - \frac{\ve(0) - q_1}{\kappa}\,.\]

 Intrigued by this behavior, we performed several numerical simulations to understand to which extent the linear approximation captures the  behavior of the correlation functions.

\begin{figure}[ht]
	\centering
	\includegraphics[scale=0.8]{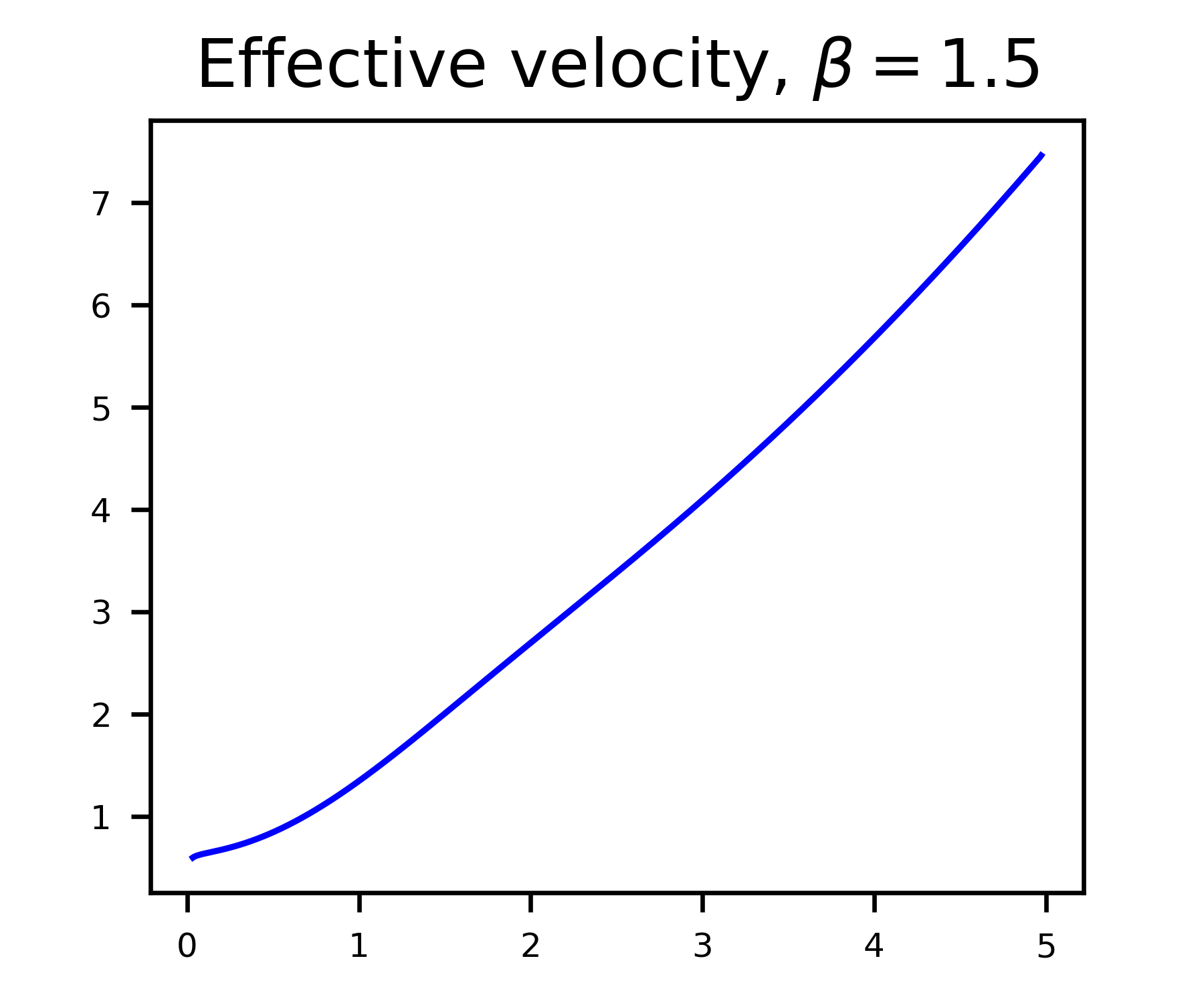}
	\caption{$\ve$ for $\beta=1.5$, $V(x) = \frac{x^2}{2}$}
	\label{fig:v_eff}
\end{figure}
\section{Numerical Results}
\label{sec:numerical_results}

In this section, we present the numerical results that we obtained, in the last part of this section we present the method that we used to numerically simulate both the classical correlation functions and the GHD predictions.

\subsection{Description of the results}
We compared the  GHD prediction of the correlation functions of the Volterra lattice with the molecular dynamics simulation (MD) for two different temperature corresponding to $\beta=1;1.5$, see figure \ref{fig:paragone1}.

In each of these cases, we have evaluated the GHD approximations (also called Landau-Lifshitz approximation) $\tS(x,t)$ \eqref{eq:sol_linear} of the correlators for all $0\leq n\leq m\leq 1$ using the numerical scheme that we describe in \ref{sec:linearized_GHD}. Their graphs are displayed in Figures \ref{fig:paragone1} as dashed black lines.
The colored lines represent the molecular dynamics simulations. According to the ballistic scaling predicted in \eqref{eq:sol_linear}, we plot $tS_{m,n}(j,t)$ as a function of $j/t$ for $t=200,400,600$. Here the values of $S_{m,n}(j,t)$ is approximated using the numerical scheme that we describe in section \ref{sec:md_simulations}. 

The agreement between the molecular dynamics simulation and the prediction of the GHD is astonishing for negative values of $\xi = \frac{x}{t}$, but for positive values of such parameter the GHD prediction does not capture the oscillation of the correlation functions. The main reason is that the relation $\xi = -\frac{\ve(w) - q_1}{\kappa_\beta}$ is not a bijection between $\R$ and $\R_+$, thus the prediction of the GHD develop a singularity at $\xi_0 = -\frac{\ve(0) - q_1}{\kappa_\beta}$ , which is exactly where the molecular dynamics simulations show an highly oscillatory behavior. {\color{black} Moreover, the point $\xi_0$ is also the minimum of $\ve(w)$.} For this reason, we believe that one has to consider some extra diffusive terms when approximating \eqref{eq:first_approx} in order to get a more precise description of the correlation functions for this model, as it is described in \cite{DeNardis2018}.  Specifically, we believe that at the point $\xi_0$ the diffusive effects are not a sub-leading correction to the transport dynamics, but they are of the same order.

{\color{black} We notice that there is also a small discrepancy on the right side of the plot between the MD simulation and the GHD prediction of  $S_{1,1}$ for $\beta=1.1$; this is not a real discrepancy, but just a finite size effect.}

\subsection{Numerical simulation}

This subsection is divided into two parts. In the first part we present the numerical scheme that we used to simulate the evolution of the Volterra lattice and to compute the correlation functions, i.e. the Molecular Dynamics simulations. In the second part, we present the numerical scheme that we used to compute the prediction of the Generalized Hydrodynamics.

\subsubsection{Molecular dynamics simulations}
\label{sec:md_simulations}

We approximate the expected value that is contained in the MD-definition of the correlations $S_{m,n}$ in equation~\eqref{eq:hamvoltN} by a standard Rounge--Kutta method (RK45), whose implementation program is written in \texttt{Python}, and can be found at \cite{Volterra_cor_software}.
First, we generate the random initial conditions distributed according to the Gibbs measure, as given by \eqref{eq:VolterraGibbs} for the i.i.d.~random variables $(a_j)_{j=1}^{2N}$, which are distributed according to a scaled $\chi^2$ random variable. We generate this random vector with \texttt{Numpy v1.23}'s native function \texttt{random.default{\textunderscore}rng().chisquare} \cite{Numpy}.
Having chosen the initial conditions in such a manner, we solve equation \eqref{eq:hamvoltN}.

For the evolution, we use a standard Rounge-Kutta algorithm of order 5 (RK45), we decided not to use the native \texttt{Scipy v1.12.0}'s algorithm \cite{scipy}, but we implemented it, in this way we could used the library \texttt{Numba} \cite{numba} to speed up the computations.

Our approximation for the expectation $S_{m,n}$ is then extracted from $10^6$ trials with independent initial conditions. Here we take the empirical mean of all trials where for each trial we also take the mean of the $N=3000$ sets of data that we generated by choosing each site on the ring for $j=0$.

We want to mention that almost all the pictures that appeared in this paper are made using the \texttt{Python} library \texttt{matplotlib} \cite{Matplotlib}.

\subsubsection{Solving linearized GHD}
\label{sec:linearized_GHD}

To numerically solve the linearized GHD equations, we use a numerical method similar to the one from \cite{Mendl2022,mazzuca2023equilibrium}. First, Eq.~\eqref{uLdplus1} is expressed in terms of Whittaker function $W_{\mu,\kappa}(z)$ \cite{dlmf}, which is readily available in \texttt{Mathematica} \cite{Mathematica}. This provides the solution to minimization problem   \eqref{eq:functional}.

Then, we use a simple finite element discretization of the $w$-dependent functions by hat functions, resulting in piece-wise linear functions on a uniform grid. After precomputing the integral operator $T$ in \eqref{eq:def_T} for such hat functions, the dressing transformation \eqref{eq:dressing} becomes a linear system of equations, which can be solved numerically. This procedure yields $[w^2]^\mathrm{dr}$, and subsequently $\sigma$ via \eqref{eq:useful} and $v_\mathrm{eff}$ via \eqref{eq:v_eff}. 

To evaluate the correlation functions in \eqref{eq:correlation}, we notice  that the delta-function in the integrand results in a parametrized curve, with the first coordinate (corresponding to $x/t$) equal to $-\frac{\ve(w) - q_1}{\kappa}$ from \eqref{eq:correlation}, and the second coordinate equal to the remaining terms in the integrand divided by the Jacobi factor $\vert \frac{\mathrm{d}}{\mathrm{d}w} \ve(w)\vert $ resulting from the delta-function.

	      {\color{black}
        \section{Conclusions and Outlooks}
        
        In this paper, we computed the GHD approximation of the correlation functions for the Volterra lattice. We were able to do so by connecting the GGE of the Volterra lattice with the classical Anti-symmetric Gaussian $\beta$ ensemble. Given the explicit expression of the GHD prediction, we noticed that it has a discontinuity for $\frac{x}{t} = \xi_0$, see \eqref{eq:correlation}. This shock is due to the fact that the effective velocity $\ve$ defined by the collision rate ansatz, see Remark \ref{rem:collision}, is an even smooth function, therefore it has a stationary point at $0$, which, in this case, is also a minimum. It would be interesting to understand if this point has some deep physical interpretation. This particular feature was already observed in \cite{Shock1} for quantum systems in the context of density ripples. We mention that in \cite{Whitham1}, the authors describe a similar phenomena using Whitham modulation theory, and it might be possible that one can apply the same idea in this context.
        
        Furthermore, we compared our findings with several numerical experiment, see Figure \ref{fig:paragone1}. We notice that the numerical correlation functions and the GHD prediction line up away from the discontinuity point $\xi_0$, but around it the actual correlations are highly oscillatory.  In order to be able to describe such region analytically, we think that it would be interesting to conduct some other experiment focusing on the $t$-dependence of the shock region.         
      
 Since the Volterra lattice is also the discretization of the KdV equation, it could be interesting to understand if one can recover the soliton gas description of the KdV obtained in \cite{Bonnemain2022} using the result in our paper.

        \appendix
        
        \section{Technical Results}
	\label{app:tech_res}
	In this appendix we prove all the technical results that we used in our paper.
	
	\subsection{Proof of Lemma \ref{lem:mean_currents}}        
	
	For the sake of the reader, we report the statement of the Lemma
        \begin{lemma}
    
            Consider the Volterra lattice \eqref{Volterra} endowed with the GGE \eqref{eq:VolterraGibbs}, and define the currents $J^{[n]}$ as in \eqref{eq:currents}, then for all fixed $n\in\N$

            \begin{equation}
            	\label{eq:local_current_app}
                \lim_{N\to\infty} \frac{1}{2N}\meanvalue{J^{[n]}}{1} = -\frac{1}{2}\int_0^\beta \partial_{t_1}\partial_{t_2} \cF_{\textrm{Volt}}(y, V + it_1x^2 + (-1)^{n+1} it_2x^{2n})\di y\,,
            \end{equation}
            where $\cF_{\textrm{Volt}}(\beta,V)$ is the free energy \eqref{eq:free_energy_volterra}.
        \end{lemma}

        To prove this lemma, we need a corollary of result from \cite{Mazzuca2024CLT} about the exponential decay of spatial correlation functions of \textit{local function}, which are functions on the phase space $\R^{2N}_+$ depending on a finite number of consecutive variables. To formally introduce this idea, we need some definitions.
        
        Given a differentiable function $F \colon \R^{2N}_+ \to \C$, we define its {\em support} as the set 
        \begin{equation}
        	\label{def:supp}
        	{\rm supp }\, F := \left\{ \ell \in \{ 1, \ldots, 2N\} \colon \ \ \   \frac{\partial F}{\partial a_\ell}\notequiv 0 \right\}
        \end{equation}
        and its  {\em diameter}  as 
        \begin{equation}
        	\label{diameter}
        	{\rm diam} \left({\rm supp }\, F\right) := \sup_{i, j \in {\rm supp}\, F} \td_{2N}(i,j) + 1 , 
        \end{equation}
        where $\td_{k(i,j)}$ is the  {\em periodic distance}  
        \begin{equation}
        	\label{p.dist}
        	\td_{k(i,j)} := \min \left( |i-j|, \ k - |i-j| \right) . 
        \end{equation}
        Note that $0\leq \td_{2N}(i,j) \leq N$.
        
        We say that a function $F$ is \textit{local} if ${\rm diam} \left({\rm supp }\, F\right)$ is uniformly bounded in $N$, i.e. there exists a constant $\tc\in \N$ such that ${\rm diam} \left({\rm supp }\, F\right) \leq \tc$, and $\tc$ is independent of $N$. 
        
       Another important class of functions that we consider are the so-called \textit{cyclic} functions, which are a class of function invariant under left or right shift of the variables. More specifically, for any $\ell \in \Z$,  and $\bx=(x_1, x_2, \ldots, x_{2N})\in \mathbb{R}^{2N}_+$  we define the {\em cyclic shift of order $\ell$} as the map  
       \begin{equation}
       	\label{shift}
       	S_\ell \colon \R^{2N} \to \R^{2N}, \qquad (S_\ell x)_j := x_{((j+\ell-1)\mod 2N) +1} . 
       \end{equation}
       For example $S_1$ and $S_{-1}$ are  the left respectively right shifts:
       $$
       S_1(x_1, x_2, \ldots, x_{2N}) := (x_2, \ldots, x_{2N},  x_{1}), \qquad
       S_{-1}(x_1, x_2, \ldots, x_{2N}) := (x_{2N}, x_1, \ldots, x_{2N-1}).
       $$
       One can immediately  check that for  any $\ell, \ell' \in \Z$:
       \begin{equation}
       	\label{prop:shift}
       	S_{\ell} \circ S_{\ell'} = S_{\ell + \ell'}, \qquad S_{\ell}^{-1} = S_{- \ell} , \qquad S_0 = \uno , \qquad
       	S_{\ell + 2N } = S_{\ell} . 
       \end{equation}
       Consider now a  function $H\colon \R^{2N}_+ \to \C$;  we denote  by $S_\ell H\colon \R^{2N}_+ \to \C$  the function \begin{equation}
       		\label{cyc.func}
       		(S_\ell H)(\ba) := H(S_\ell \ba) , \qquad \forall  \ba \in \R^{2N}_+\,.
       	\end{equation}
       	Clearly $S_\ell$ is a linear operator.
       	We can now define cyclic functions:
       	\begin{definition}[Cyclic functions]
       		\label{def:cyclic}
       		A function $H\colon \R^{2N}_+ \to \C$ is called {\em cyclic} if $S_1 H = H$.
       	\end{definition}
       	It is easy to construct cyclic functions as follows: given a function $h \colon \R^{2N}_+ \to \C$  we define the  new function $H$ by
       	\begin{equation}
       		\label{seed}
       		H(\ba)  := \sum_{\ell = 0}^{2N-1} (S_{\ell} h)(\ba) .
       	\end{equation}
       	$H$ is clearly  cyclic and  we say  that  $H$ is {\em generated } by $h$, we remark that these definition were introduced in this context in \cite{Giorgilli2014,Grava2020}. 
              	
          \begin{remark}
          	\label{rem:cyclic_prop}
          According to the previous definition, the conserved field $Q^{[n]}$  and the currents $J^{[n]}$ of the Volterra lattice are cyclic; furthermore, their seed are local functions.  Given these properties, we call these seeds $Q^{[n]}_j$ \textit{local conserved fields} and $J^{[n]}_j$\textit{local currents} .
          \end{remark}

       	Given these definitions, we can state the following corollary:
        
\begin{corollary} [Decay of correlations]
    \label{cor:decay}
        Consider the Volterra lattice \eqref{Volterra} endowed with the GGE \eqref{eq:VolterraGibbs}, and let $I,J : \R^{2N}_+ \to \R $ be two local functions with the same support of diameter $k$. Assume that they are integrable with respect to the GGE \eqref{eq:VolterraGibbs}. Write $2N=kM+\ell$, and let $j\in \{1,\ldots,M\}$. Then, there exists some $0<\mu<1$ such that
    $$ \E_{1}\left[I(\ba)S_jJ(\ba)\right]-\E_{1}\left[I(\ba)\right]\E_{1}\left[S_jJ(\ba)\right]=O(\mu^{\td_M(j,0)})\,.$$
\end{corollary}

With the previous corollary, we can prove Lemma \ref{lem:mean_currents} 
     \begin{proof}[Proof of Lemma \ref{lem:mean_currents}]

         First, we notice that in view of the cyclic property of the total currents

         \begin{equation}
             \lim_{N\to\infty}\frac{1}{2N}\meanvalue{J^{[n]}}{1} = \meanvalue{J^{[n]}_1}{1}\,. 
         \end{equation}
         Moreover, we have the following chain of equality

         \begin{equation}
         \label{eq:partial_beta}
             \partial_\beta \meanvalue{J^{[n]}_1}{1} = \cov{J_1^{[n]}}{Q^{[0]}} = \sum_{j=1}^{2N} \cov{J_1^{[n]}}{Q_j^{[0]}}\,.
         \end{equation}
         Assume that the following equality holds

         \begin{equation}
         \label{eq:exchange}
            \lim_{N\to\infty} \cov{J_1^{[n]}}{Q^{[m]}_1} = \lim_{N\to\infty} \cov{Q^{[n]}_1}{J_{2N-j+2}^{[m]}}\,,
         \end{equation}
         then we can recast \eqref{eq:partial_beta} as

         \begin{equation}
             \partial_\beta \meanvalue{J^{[n]}_1}{1} = \sum_{j=1}^{2N} \cov{Q_1^{[n]}}{J_{2N-j+2}^{[0]}} \stackrel{\text{Remark} \ref{rem:trick_rem}}{=} -\frac{1}{2} \sum_{j=1}^{2N} \cov{Q_1^{[n]}}{Q_j^{[1]}}\,.
         \end{equation}
        Furthermore, we notice that $\lim_{\beta \to 0}\meanvalue{J^{[n]}_1}{1} = 0$, thus, applying Corollary \ref{cor:mean_cor_gibbs}, we deduce the following

        \begin{equation}
            \lim_{N\to\infty} \frac{1}{2N}\meanvalue{J^{[n]}}{1} =-\frac{1}{2} \int_0^\beta \partial_{t_1} \partial_{t_2} \cF_{\textrm{Volt}}(x, V+it_1x^2 + (-1)^{n+1}it_2x^{2n})\di x\,.
        \end{equation}

         So, we have just to show that \eqref{eq:exchange} holds. Consider the following chain of equality

         \begin{equation}
             \begin{split}
                 \cov{J_{j-1}^{[n]}(t) -J_{j}^{[n]}(t)}{Q_{1}^{[m]}(0)} & = -\frac{\di}{\di t} \cov{Q_{j}^{[n]}(t)}{Q_{1}^{[m]}(0)} \\
                 & = -\frac{\di}{\di t} \cov{Q_{j}^{[n]}(0)}{Q_{1}^{[m]}(-t)} = -\frac{\di}{\di t} \cov{Q_{1}^{[n]}(0)}{Q_{2N-j+2}^{[m]}(-t)}\\
                 & = \ \cov{Q_1^{[n]}(0)}{J_{2N-j+2}^{[m]}(-t) - J_{2N-j+1}^{[m]}(-t) }\,.
             \end{split}
         \end{equation}
     Setting $\partial_jf(j) = f(j) - f(j-1)$, we proved that

     \begin{equation}
         \partial_j\left(\cov{Q_1^{[n]}(0)}{J_{2N-j+2}^{[m]}(0)} - \cov{Q_1^{[m]}(0)}{J_{j}^{[n]}(0)}\right) = 0\,.
     \end{equation}
     Thus, $\cov{Q_1^{[n]}(0)}{J_{2N-j+2}^{[m]}(0)} - \cov{Q_1^{[m]}(0)}{J_{j}^{[m]}(0)}$ is independent of $j$, but all the function involved are local function, so we can apply Corollary \ref{cor:decay} to show \eqref{eq:exchange} holds. So we conclude.
     \end{proof}
     
     \subsection{Proof of Proposition \ref{prop:dress_prop}}
     
     For the sake of the reader, we rewrite the statement that we want to prove
     
     \begin{proposition}
	Consider the measure $\varrho$ defined as the unique minizier of \eqref{eq:var_no_beta}, the operator $T$ defined in \eqref{eq:def_T} and the dressing operator \eqref{eq:dressing}. Then the following holds true
	
	\begin{enumerate}
		\item for any function $f$
		\begin{equation}
			(1-\varrho T)^{-1}[\varrho f] = \varrho(1- T\varrho)^{-1}[f]\,,
		\end{equation}
		
		\item  for any function $f,g$ 
		\begin{equation}
		\label{eq:prop2_app}
			\jap{(1-\varrho T )^{-1}[f] g} = \jap{f \dr{g}}\,,
		\end{equation}
		\item for any variable $\odot$
		
		\begin{equation}
		\label{eq:prop3_app}
			\partial_\odot \partial_\mu \varrho = (1-\varrho T)^{-1} \partial_\odot \varrho (1- T\varrho)^{-1}[1]
		\end{equation}
		
		\item for any function $f$
		
		\begin{equation}
			\partial_\mu\jap{\sigma f} = \jap{\sigma \dr{1}\dr{f}}\,.
		\end{equation}
		
		\item Consider the perturbed potential $V(x) \to V(x) +(-1)^{n+1}it_nx^{2n}$ then 
		
		 \begin{equation}
			\partial_{t_n} \mu(\beta, V(x) + (-1)^{n+1} it_n x^{2n})_{\vert_{t_n=0}} = 2iq_n\,.
		\end{equation} 
		
		\item For any function $\psi$
		
		\begin{equation}
			\partial_{t_n}\left( (1-T\varrho(V+ (-1)^{n+1}it_n\lambda^{2n}))^{-1}(\psi)\right)_{\vert_{t_n=0}} = (1-T\varrho)^{-1}\left(T\partial_{t_n}\varrho(V+ (-1)^{n+1}it_n\lambda^{2n})_{\vert_{t_n=0}}\dr{\psi}\right)\,.
		\end{equation}
	\end{enumerate}
\end{proposition}

\begin{proof}
	(1) It is equivalent to prove that 
	\begin{equation}
		\varrho f = (1-\varrho T)\left[\varrho(1- T\varrho)^{-1}[f]\right]\,,
	\end{equation}
	which follows from straightforward computations.
	
	(2) For any function $y,h$ the following equality holds
	
	\begin{equation}
		\jap{y(1-T\varrho)h} = \jap{h(1-\varrho T)y}\,,
	\end{equation}
	which leads to \eqref{eq:prop2_app} setting $y = (1-\varrho T)^{-1} f,\, h = (1- T\varrho)^{-1} g $.
	
	(3) By Differentiating the equality 
	\begin{equation}
		(1-\varrho T)\sigma = \varrho\,,
	\end{equation}
	we deduce
	
	\begin{equation}
		(1-\varrho T) \partial_{\odot} \sigma = \partial_\odot \varrho (1 + T\sigma)\,.
	\end{equation}
	
	If we can prove that $(1+T\sigma) = (1 - T\varrho)^{-1}$ we conclude.
	
	\begin{equation}
		\begin{split}
			1+T\sigma = 1 + T \varrho(1-T\varrho)^{-1}[1] = (1-T\varrho)(1-T\varrho)^{-1}[1]  + T \varrho(1-T\varrho)^{-1}[1] = (1-T\varrho)^{-1}[1]\,,
		\end{split}
	\end{equation}
	so we conclude.
	
	(4) From the previous relation we deduce the following chain of equality
	
	\begin{equation}
		\partial_{\mu}\jap{\sigma f } \stackrel{\eqref{eq:prop3_app}}{=} \jap{(1-\varrho T)^{-1}\sigma (1- T \varrho)^{-1}[1] f} \stackrel{\eqref{eq:prop2_app}}{=} \jap{\sigma \dr{1}\dr{f}}
	\end{equation}
	
	(5) By differentiating the Euler-Lagrange equation \eqref{eq:euler_lagrange}  with respect to $t_n$, we deduce that
	
	\begin{equation}
		-\int_{\R_+}\ln(\vert x^2-y^2 \vert)\partial_{t_n}\varrho(y)_{\vert_{t_n=0}}\di y  +(-1)^{n+1} 2ix^{2n} + \frac{\partial_{V}\varrho_{\vert_{t_n=0}}}{ \varrho_{\vert_{t_n=0}}} + (-1)^{n}\partial_{V} \mu(\beta,V)_{\vert_{t_n=0}} = 0\,,
	\end{equation}
	Testing the previous variational equation against $\sigma$ we deduce,
	
	\begin{equation}
		(-1)^{n+1}2i\jap{w^{2n}\sigma} - \jap{\sigma T\partial_{t_n}\varrho_{\vert_{t_n=0}}} +\jap{(1+T\sigma)\partial_{V}\varrho_{\vert_{t_n=0}}} = \jap{\partial_{t_n} \mu(\beta,V)_{\vert_{t_n=0}}\sigma} = (-1)^{n+1}\partial_{t_n} \mu(\beta,V)_{\vert_{t_n=0}} \kappa^{-1} \,,
	\end{equation}
	which leads to the conclusion.
	
	(6) From \eqref{eq:dressing} we take the derivative with respect to $t_n$, getting that
	
	\begin{equation}
		(1-T\varrho)\partial_{t_n}\dr{\psi} = T\partial_{t_m}\varrho\dr{\psi}\,,
	\end{equation}
	which leads to the conclusion. Here, we have omitted the explicit dependence of $\varrho$ from the potential, and the evaluation at $0$.
\end{proof}

\subsection{Proof of Theorem \ref{thm:LDP_antisym}}
For reader convenience, we report the statement of the Theorem.
\begin{theorem}

	Consider the functional $\cF_{\textrm{AG}}(\beta,V)[\rho]$ defined as 
	\begin{equation}
		\label{eq:functional_app}
		\begin{split}
			\cF_{\textrm{AG}}(\beta,V)[\rho] &		= -\frac{\beta}{2} \int\int_{\R_+^2}\ln(\vert x^2-y^2 \vert)  \rho(x) \rho(y)\di x \di y - \int_{\R_+}(V(ix)  +V(-ix)  + \ln(|x|)) \rho(x)\di x \\ & +\int_{\R_+}\ln( \rho(|x|)) \rho(|x|)\di x 
		\end{split}
	\end{equation}
	here $\rho(x)$ is an absolutely continuous measure with respect to the Lebesgue one, has support on the positive real line. The previous functional has a unique minimizer  $\rho_{\beta,V}(x)$, which is absolutely continuous with respect to the Lebesgue measure. In particular, $\rho_{\beta,V}(x)$ is the density of states of the Anti-symmetric $\beta$ ensemble in the high temperature regime. Furthermore, 
	
	\begin{equation}
		\fF_{\textrm{AG}}(\beta,V)  = \frac{1}{2}\cF_{\textrm{AG}}(\beta,V)[\rho_{\beta,V}]  + \frac{1}{2}\int_0^1 \ln\left(\frac{\beta}{2}x\right) dx - \frac{\ln(2)}{2}\,.
	\end{equation}
	%	 Furthermore, the following holds for any continuous and bounded function $f$
	%	
	%     	\begin{equation}
		%	\int_{\R} f(x)d\wt \rho_{\beta}^V(x) = 	\int_{\R_+}\rad{f} d\rho_{\beta}^V(x) = \int_{\R}\vert x \vert f(x) d\rho_{\beta}^V(x^2) \,.
		%\end{equation}
		%	
	\end{theorem}
	
	\begin{proof}
		
		First from the definition of Free energy and \eqref{eq:jacobian} we deduce that
		
		\begin{equation}
			\begin{split}
				\fF_{\textrm{AG}}(\beta,V) & = -\lim_{N\to\infty} \frac{1}{2N} \ln\left(Z_N^{AG}(\beta/N,V) \right) \\ & = -\lim_{N\to\infty} \frac{1}{2N} \ln\left(\zeta_N\left(\frac{\beta}{N}\right)\right) \\ & -\lim_{N\to\infty} \frac{1}{2N} \ln\left(\int_{w_1<w_2<\ldots <w_n} \prod_{j=1}^{N}w_j^{\frac{\beta}{2N} -1}e^{\sum_{j=1}^N V(iw_j) + V(-iw_j)}\prod_{1\leq j < i \leq N} \left(w_j^2 -w_i^2 \right)^{\frac{\beta}{N}}d \bw\right)\,.
			\end{split}
		\end{equation}
		
		The first term can be explicitly computed as
		
		\begin{equation}
			\label{eq:free_base}
			-\lim_{N\to\infty} \frac{1}{2N} \ln\left(\zeta_N\left(\frac{\beta}{N}\right)\right) \stackrel{\eqref{eq:zeta}}{=} \frac{1}{2}\left( \int_0^1 \ln\left(\frac{\beta}{2}x\right) dx  - \ln(2)\right)\,.
		\end{equation}
		For the second term, we can apply theorem \cite[Theorem 1.1]{Zelada2019}, to deduce that 
		
		\begin{equation}
			\label{eq:simpify}
			\begin{split}
				-\lim_{N\to\infty}& \frac{1}{2N} \ln\left(\int_{w_1<w_2<\ldots <w_n} \prod_{j=1}^{N}w_j^{\frac{\beta}{2N} -1}e^{\sum_{j=1}^N V(iw_j) + V(-iw_j)}\prod_{1\leq j < i \leq N} \left(w_j^2 -w_i^2 \right)^{\frac{\beta}{N}}d \bw\right) \\ & = \frac{1}{2}\min_{\wt \rho\in \mathcal{P}(\R_+)
				}\cF_{\textrm{AG}}(\beta,V)[\wt \rho]  \,,
			\end{split}
		\end{equation}
		where  $\mathcal{P}(\R_+)$ is the space of probability measure with support on the positive real line and 
		\begin{equation}
			\begin{split}
				\wt  \cF_{\textrm{AG}}(\beta,V)[\wt \rho]&   = 	 -\frac{\beta}{2} \int\int_{\R_+^2}\ln(\vert x^2-y^2 \vert) \rho(x) \rho(y)\di x \di y - \int_{\R_+}(V(ix) +  V(-ix) - \ln(x)) \rho(x)\di x  \\ & +\int_{\R_+}\ln(\rho(x))\rho(x)\di x  \,.
			\end{split}
		\end{equation}
		
		Combining the previous expression with \eqref{eq:free_base}, we deduce the claim.
	\end{proof}

}

        \section*{Data Availability statement}
        All the \texttt{pyhton} and \texttt{mathematica} codes used for the numerical simulation of this paper are freely available at \cite{densityrepo}.
        
        \section*{Conflict of Interest Statement}
        
         The author certify that they have NO affiliations with or involvement in any
organization or entity with any financial interest (such as honoraria; educational grants; participation in speakers’ bureaus;
membership, employment, consultancies, stock ownership, or other equity interest; and expert testimony or patent-licensing
arrangements), or non-financial interest (such as personal or professional relationships, affiliations, knowledge or beliefs) in
the subject matter or materials discussed in this manuscript.
	\bibliographystyle{siam}
	\bibliography{mybib}

\end{document}